\documentclass[sigplan,screen]{acmart}
\usepackage[utf8]{inputenc}
\usepackage{graphicx}
\usepackage{float}
\usepackage{amsthm}
\usepackage{amsmath}
\usepackage{hyperref}
\usepackage{multicol}
\usepackage{paralist}
\usepackage{tabularx}
\usepackage{listings}
\usepackage{enumitem}

\newcommand{\mypara}[1]{\noindent\textbf{#1.}}
\newcommand{\del}{\texttt{delete}}
\newcommand{\delargs}{\texttt{delete(key)}}
\newcommand{\searchargs}{\texttt{search(key, target)}}
\newcommand{\insargs}{\texttt{insert(key, val)}}
\newcommand{\findargs}{\texttt{find(key)}}
\newcommand{\key}{\texttt{key}}
\newcommand{\search}{\texttt{search}}
\newcommand{\find}{\texttt{find}}
\newcommand{\ins}{\texttt{insert}}
\newcommand{\val}{\texttt{val}}

\newcommand{\var}[1]{\texttt{#1}}
\newcommand{\tagged}{\texttt{TaggedInternal}}
\newcommand{\fixtagged}{\texttt{fixTagged}}
\newcommand{\fixunderfull}{\texttt{fixUnderfull}}
\newcommand{\pathinfo}{\texttt{PathInfo}}

\newcommand{\elimrec}{\texttt{ElimRecord}}
\newcommand{\lockOrElim}{\texttt{lockOrElim}}
\newcommand{\searchleaf}{\texttt{searchLeaf}}

\newcommand{\lfab}{LF-ABtree}
\newcommand{\occ}{OCC-ABtree}
\newcommand{\elim}{Elim-ABtree}
\newcommand{\pmocc}{p-OCC-ABtree}
\newcommand{\pmelim}{p-Elim-ABtree}

\copyrightyear{2022}
\acmYear{2022}
\setcopyright{rightsretained}
\acmConference[PPoPP '22]{27th ACM SIGPLAN Symposium on Principles and Practice of Parallel Programming}{April 2--6, 2022}{Seoul, Republic of Korea}
\acmBooktitle{27th ACM SIGPLAN Symposium on Principles and Practice of Parallel Programming (PPoPP '22), April 2--6, 2022, Seoul, Republic of Korea}
\acmDOI{10.1145/3503221.3508441}
\acmISBN{978-1-4503-9204-4/22/02}

\begin{document}

\title{Elimination (a,b)-trees with fast, durable updates}

\author{Anubhav Srivastava}
\email{anubhav.srivastava@uwaterloo.ca}
\affiliation{%
  \institution{University of Waterloo}
  \city{Waterloo}
  \country{Canada}
}

\author{Trevor Brown}
\authornote{Corresponding author}
\email{trevor.brown@uwaterloo.ca}
\affiliation{%
  \institution{University of Waterloo}
  \city{Waterloo}
  \country{Canada}
}

\renewcommand{\shortauthors}{Srivastava and Brown}



 %\vspace{-2mm}}

\lstset{
    escapechar=@,
    firstnumber=last,
    belowcaptionskip=1\baselineskip,
    breaklines=true,
    frame=single, % It seems like it might be more readable without the frame
    xleftmargin=\parindent,
    language=C++,
    showstringspaces=false,
    basicstyle=\scriptsize\ttfamily, %\scriptsize\ttfamily,
    keywordstyle=\bfseries\color{green!40!black},
    commentstyle=\itshape\color{gray!60!black},
    identifierstyle=\color{black!50!black},
    stringstyle=\color{orange},
    numbers=left,                    % where to put the line-numbers; possible values are (none, left, right)
    numbersep=2mm,                   % how far the line-numbers are from the code
    xleftmargin=3mm,
    numberstyle=\tiny\color{black}, % the style that is used for the line-numbers
    morekeywords={procedure,then,word,type,xbegin,xabort,xend,is,in,to,until},
    % literate    ={--}{\textendash}1 {---}{\textemdash}2,
}

%%%%% Trevor Added %%%%%%%%%
% \definecolor{pdfbgcolor}{RGB}{180,180,180}
% \pagecolor{pdfbgcolor}
%%%%%%%%%%%%%%%%%%%%%%%%%%%%


\begin{abstract}
Many concurrent dictionary implementations are designed and optimized for read-mostly workloads with uniformly distributed keys, and often perform poorly on update-heavy workloads.
In this work, we first present a concurrent (a,b)-tree, the \occ, which outperforms its fastest competitor by up to 2x on uniform update-heavy workloads, and is competitive on other workloads.
% typically matches the performance of its fastest competitor on read-mostly workloads.
We then turn our attention to \textit{skewed} update-heavy workloads (which feature many inserts/deletes on the same key) and introduce the \elim, which features a new optimization called publishing elimination.
In publishing elimination, concurrent inserts and deletes to a key are reordered to \textit{eliminate} them.
This reduces the number of writes in the data structure. The \elim\ achieves up to 2.5x the performance of its fastest competitor (including the \occ).
The \occ\ and \elim\ are linearizable. We also introduce durable linearizable versions\footnote{Technically, we show that our data structures satisfy a stronger correctness condition called \textit{strict linearizability}~\cite{StrictLinearizability}.} for systems with Intel Optane DCPMM non-volatile main memory that are nearly as fast.
% also performant in the \textit{persistent} memory setting, with no extra recovery logic (only cache line flushes).
\end{abstract}

\maketitle

\section{Introduction}
The (ordered) dictionary is one of the most fundamental abstract data types.
It stores a set of keys, each of which has an associated value, and provides operations to insert a key and value, remove a key, and find the value associated with a key.
Sometimes dictionaries also support predecessor, successor, and range query operations.

% - question of how to scale reads is well studied
% - we are interested in the hard question of how to scale update heavy workloads
% - want to maintain good performance in commonly studied workloads while surpassing existing approaches in challenging 

Concurrent dictionary implementations in the literature typically focus on maximizing performance under low contention read-mostly workloads, % with a uniform access distribution,
with less attention paid to performance under update-heavy workloads and high contention workloads.
% These are workloads where either most operations are read-only, or updates are performed on uniformly distributed keys (and thus rarely encounter contention).
%There has also been some work on skewed workloads.
In this paper, we study the question of how to scale these challenging workloads, ideally without sacrificing performance in the read-mostly workload.
%our goal is to 
%obtain good performance in these commonly-studied workloads, while \textit{surpassing} existing approaches in other workloads. In particular, we also study update-heavy and high contention workloads.
Update-heavy workloads are particularly difficult to scale when there is a lot of memory contention.
% A particularly challenging update-heavy workload for concurrent data structures is the high-contention case.
% We also want to consider the case of skewed access distributions, especially for update-heavy workloads (skewed read-mostly workloads are not as interesting; the most common tree paths simply end up in cache). 
% We also want to consider high contention workloads.
%There are two typical ways to simulate high contention.
%The first is to use very small data structures. %However, many real-life data structures contain millions of keys.
%The second, arguably more realistic, method is to use a skewed access distribution.
% But, we would like to study high contention workloads in \textit{large} data structures.
% So, we study update-heavy workloads in which threads access keys according to a skewed probability distribution.
To generate high contention, we study Zipfian access distributions, in which the frequency of a key being accessed is inversely proportional to its rank.
That is, the $k$th most frequent key is requested with probability proportional to $1/k^s$, where $s$ is a parameter controlling the skew of the distribution.

% Dictionaries are often implemented using trees.
% Concurrent binary search trees are well-studied and have been heavily optimized, but even balanced binary search trees suffer from cache misses caused by long paths to their leaves.
% In contrast, B-trees use internal nodes that contain more than one key to achieve a shallower depth.
% A B-tree of order $b$ has between $b/2$ and $b$ keys in each internal node, meaning that the height of the tree is at most $\log_{b/2} K$, where $K$ of the number of keys in the tree.
% This is significantly smaller than the $\log_2 K$ height of BSTs for large $K$, and practically leads to faster searching in B-trees.
% %The time taken to search an internal node in a B-tree is negligible since $b$ is usually chosen so a single node fits in a cache line. 
% Moreover, %since nodes only need to be sizes between $b/2$ and $b$ means 
% rebalancing in B-trees occurs less frequently than in balanced BSTs.

The advantages of concurrent B-trees over binary search trees, including better cache locality, %and more straightforward rebalancing algorithms,
are well known.
Our new data structures presented in this paper are (a,b)-trees, which are a variant of B-trees that allow between $a$ and $b$ keys per node (for $a \leq b/2$).
Our trees are based on the (non-concurrent) \textit{relaxed} (a,b)-tree of Larsen and Fagerberg~\cite{Larsen}.
Relaxed (a,b)-trees are more concurrency friendly than B-trees.
They break insert and delete operations, and any subsequent rebalancing, into multiple \textit{sub-operations} (each of which modifies at most four nodes).
As long as each \textit{sub-operation} is atomic, the relaxed (a,b)-tree's structure and balance properties are maintained.
Implementing these sub-operations atomically requires less synchronization than implementing traditional (sequential) B-tree operations atomically, since B-tree operations must sometimes rebalance an entire root-to-leaf path.

Relaxed (a,b)-trees have been implemented in a concurrent setting before~\cite{TrevorThesis,TrevorHTMTemplate,EBRRangeQueries}, but the overheads of existing implementations are high, and they perform poorly in update-heavy workloads.
Our first new data structure, an optimistic concurrency control (a,b)-tree (\occ), uses mostly known techniques to avoid the main sources of overhead in those implementations: unnecessary node copying and key sorting in leaves, and various overheads introduced by lock-free synchronization primitives.

The main challenge of creating a \textit{concurrent} relaxed (a,b)-tree is guaranteeing that sub-operations occur atomically, and that searches are correct.
The \occ\ uses fine-grained versioned locks to achieve the former, and version based validation in leaf nodes for the latter.
This locking technique is somewhat similar to OPTIK~\cite{DGT} and the optimistic validation of the AVL tree of Bronson et~al.~\cite{BCCO}.
% OCC involves the use of version numbers to validate nodes: the version number of a node is incremented before and after every update (from an even number to an odd number, then back to an even number).
% So, seeing the same even version number at two different reads guarantees that the node was not modified between the reads.
%The resulting tree,
As our experiments show, the \occ\ outperforms many state-of-the-art data structures on \textit{both} read-mostly and update-heavy workloads.
However, like its competitors, its performance degrades as contention increases.

% At first glance, it might not seem like a tree with \textit{fat nodes} (containing many keys) is the right data structure to use if one wants to achieve good performance under high contention.
% Although fat nodes promote better memory locality in low-contention workloads, they could easily become a liability under high contention, especially if multiple \textit{hot keys} (which are updated frequently) are located in the same node.
% However, using fat nodes minimizes the height of the tree, which is crucial for obtaining high performance in low-contention and read-only workloads.
% So, we must optimize for high contention in trees with fat nodes.

To optimize for high-contention workloads, we take inspiration from another data structure that tackles extremely high contention: \textit{elimination stacks}~\cite{StackElimination}.
In an elimination stack, whenever a thread experiences contention while accessing the stack, it attempts to synchronize \textit{directly} with another thread performing the opposite operation (push/pop) to complete both operations \textit{without} accessing the stack.
%Flat combining is another (more general) approach 

Our second new data structure, the \elim, uses a new type of elimination called \textit{publishing elimination}.
This is a primary contribution in this work.
In publishing elimination, threads that modify a leaf place a \textit{record} of their modification in the leaf itself.
Other threads can then use this record to return from their operation without having to modify the data structure.
In traditional elimination, \textit{pairs} of threads rendezvous and eliminate each others' operations.
In publishing elimination, \textit{many} threads can use a \textit{single record} in a leaf to eliminate their own operations.
The \elim\ is significantly faster than the \occ\ (and prior work) in high contention workloads.

% Our publishing elimination can also be viewed as a specialization of flat combining~\cite{FlatCombining} that is heavily optimized for our data structures.
% These optimizations rely on subtle data structure-specific linearization arguments.

% , the threads in publishing elimination do not have to synchronize with each other, making publishing elimination much faster than traditional elimination and flat combining.
% Our second algorithm, the \elim, adapts the idea of \textit{elimination} from stacks for use in trees.
% Our \elim\ uses a novel elimination technique in which every successful update to a leaf node records its result in an \textbf{elimination record}.
% Concurrent inserts and deletes can use the elimination record to eliminate themselves without modifying the tree. All operations are still linearizable.

Publishing elimination is especially enticing in systems with Intel Optane (DCPMM) persistent main memory, because fewer \textit{flushes} and high-latency \textit{fence} instructions are needed.
% \begin{compactenum}
    % \item Writes must be followed by \textit{flushes} and high-latency \textit{fence} instructions
    % \item Write bandwidth is much lower than read bandwidth when the set of accessed addresses exceeds the size of the memory controller's cache
    % \item Memory experiences wear (writes damage memory)
% \end{compactenum}
We present durably linearizable~\cite{DurableLinearizability} implementations of the \occ\ and \elim. 
This requires minor modifications to the code to add flushing and fencing as appropriate to ensure that % because
each update appears to occur atomically in persistent memory. %to both trees appear \textit{atomically} in memory, and thus the recovery procedure after a crash is trivial.
The resulting persistent trees are only slightly slower than their volatile counterparts, offering persistence at nearly the speed of in-memory computing. %as shown in our experiments.
% We perform experiments to compare our implementations to the state of the art in concurrent ordered dictionaries, and the results show that our new algorithms are extremely fast.
% Our fastest algorithm, \pubelim, outperforms the next fastest competitor by up to 2x. %, and is only occasionally outperformed by another algorithm---by
% up to 27\% in read-only workloads, and up to .

% The first is that updates on a key often synchronize on an entire node, so operations on other keys in the same node must \textit{wait} for the update on the first key to complete (or \textit{help} it complete in a lock-free tree).
% The second is that the rebalancing operations of B-trees and (a,b)-trees are more complicated than those of BSTs. Despite this, we show that concurrent (a,b)-trees can outperform many state-of-the-art concurrent dictionary implementations on both low-contention and high-contention workloads.

% \mypara{Contributions}
\textbf{Contributions.}
    (1) We present two novel algorithms: \occ\ and \elim\ which outperform the state-of-the-art in many workloads.
    (2) We introduce a novel publishing elimination algorithm that is optimized for our data structures.
    (3) We add persistence to our trees, and present experiments that show the overhead of persistence is low. %they are only slightly slower than their volatile counterparts.
    (4) Our algorithms have strong theoretical properties: deadlock-freedom, linearizability (for the volatile trees) and durable linearizability (for the persistent trees), and can be modified to guarantee logarithmic height bounds with some overhead.

\section{Related work}\label{sec:related}
%There are many concurrent ordered set implementations, but 
% Most ordered set data structures are optimized for, and experimentally evaluated under, uniform distributions.
% A uniform distribution implies low contention since, even under update-only workloads, it is unlikely that two threads are accessing the same key (unless the key range is extremely small).
% Our focus is on high-contention workloads, which implies that there are many updates and the distribution of updates is concentrated on a few keys. However, we do not want to sacrifice performance in the low-contention case. To this end, our data structures use optimistic concurrency control (OCC), efficient locks (with constant remote memory reference (RMR) complexity), and a novel form of elimination.

We briefly survey the state-of-the-art in concurrent ordered dictionaries and contrast with our techniques. We experimentally compare with the bolded data structures.

%\noindent\textbf{Binary search trees. } %
\vspace{1mm}\mypara{Binary search trees}
Ellen et~al.~\cite{Ellen} introduced the first lock-free external BST.
Searches are implemented the same way as in a sequential BST.
An update operation searches for a target node to modify, then synchronizes by \textit{flagging} or \textit{marking} nodes to indicate that they will be modified.
Other updates that encounter these flags or marks will \textit{help} the operation complete, guaranteeing lock-free progress.
Natarajan and Mittal~\cite{NM} improved upon this design by flagging/marking \textit{edges} instead of nodes, and reducing the amount of memory allocated per update operation (\textbf{NM14}).

Bronson et~al.~\cite{BCCO} propose a partially external balanced BST (\textbf{BCCO10}) that uses optimistic concurrency control to synchronize threads.
They introduce a complex hand-over-hand version number based validation technique to implement fast searches.
Our synchronization technique for updates is somewhat similar to BCCO10, but our searches avoid the complexity of Bronson's hand-over-hand validation transactions.
BCCO10 has previously been shown to be the fastest concurrent BST in search-dominated workloads~\cite{GettingToTheRoot}.

David et~al.~\cite{DGT} propose a set of rules for optimizing concurrent data structures.
They apply these rules to design a straightforward, efficient lock-based external BST (\textbf{DGT15}).
%Their algorithm uses optimistic concurrency control techniques.
%Their algorithm is concise and elegant.

Brown et~al.~\cite{LLXPrimitives} introduced wait-free synchronization primitives (\var{LLX} and \var{SCX}), used them to implement a \textit{template} for lock-free trees, and used the template to produce a (balanced) chromatic tree~\cite{TrevorTemplate}.
Several other concurrent BST algorithms have also been proposed~\cite{Howley, WaitFreeRBTrees, Ramachandran}.

\vspace{1mm}\mypara{B-tree variants}
% B-trees have also been mostly evaluated on uniform distributions.
Brown used the aforementioned template to design a lock-free (a,b)-tree (\textbf{\lfab})%, which is a concurrency-friendly variant of an (a,b)-tree
~\cite{TrevorThesis}, based on the same relaxed (a,b)-tree as our \occ~\cite{Larsen}.
%The \lfab\ decouples structural operations (splitting and joining) from logical operations (insert, delete, and find) in the same way as our \occ. %to increase performance.
%These structural operations were originally described by Larsen and Fagerberg in~\cite{Larsen}.
Update operations take a read-copy-update approach: inserting or deleting a key involves replacing a tree node with a new copy.
The \lfab\ has been shown to be substantially faster than NM14 and BCCO10, which are among the fastest BSTs~\cite{C-IST}.
% Our algorithms improve upon the \lfab\ by replacing the relatively slow read-copy-update approach with in-place modifications.
% In part, this involved replacing the \var{LLX} and \var{SCX} primitives with optimistic concurrency control techniques, modifying leaf nodes to use an \textit{unsorted} array of keys, and changing the search algorithm to account for possible key changes while reading a node.
As our experiments show, our trees significantly outperform the \lfab\ in many workloads. % It is a lock-free implementation of the relaxed (a,b)-tree presented in [Larsen] and is competitive with the BSTs mentioned above. Keys are stored in an unsorted array in the leaves, making inserts and deletes very fast. The relaxed (a,b)-tree decouples the rebalancing operations and the logical operations (insert, delete, and search) to increase performance.

% The CTrie is not ordered, don't even say anything about it since that opens the door to talking about hash tables and stuff too

Braginsky and Petrank introduced the first lock-free, linearizable B+tree~\cite{PetrankBTree}.
Each node contains a lock-free linked list of entries implemented using arrays.
Each entry is a key-value pair stored in the same word. 

The Bw-Tree is a lock-free variant of a B+tree that is designed to achieve high performance under realistic workloads~\cite{Bw-Tree}. %achieves good performance on real-world data, and additionally targets the use case where there is too much data to fit in memory. It uses a mapping table to map logical nodes to physical pages. This is motivated by the desire to avoid modifying pages directly and instead prepend a delta that represents a record modification (e.g. insert) or page operation (e.g. splitting a page) to a base page. This approach reduces cache invalidations for searches searching the page. The data structure sizes we consider in this paper all fit in memory, but we compare against the Bw-Tree in Section~\ref{} nonetheless.
Many of the design decisions made in the Bw-Tree are focused on workloads that do not fit in memory, and incur significant overhead. %, and these optimizations appear to reduce performance in workloads that fit in memory.
Our experiments include an optimized variant of the Bw-Tree called the \textbf{OpenBw-Tree}~\cite{OpenBw-Tree}. %, which is a faster variant of~\cite{original Bw-Tree paper}.

The BzTree~\cite{BzTree} simplifies the implementation of the Bw-Tree by using a multi-word compare-and-swap (MwCAS), and results in the paper suggest it is faster than the BwTree.
% Unfortunately, when we included it in our benchmarks, it did not pass our validation checks, and we were unable to resolve the issues with the authors.
%validation failures which we were unable to resolve, even with the help of the authors.
%
Guerraroui et al. introduced a faster MwCAS algorithm and used it to accelerate the BzTree.
The BzTree can also be made persistent by using a persistent MwCAS.
% The authors highlight the simplicity of their code (3000 lines of code and low cyclomatic complexity).

% In their thesis~\cite{BTreesNUMA}, McKenzie explores a number of modifications to a concurrent B-link tree~\cite{MalbrainBlinkTree}, with the goal of improving performance in systems with non-uniform memory access (NUMA).
% The most relevant modification is the idea of using NUMA-aware reader/writer locks.
% Our algorithms could possibly benefit from more sophisticated locks, but there may be a tradeoff with increased space overhead. %%Although we have not tried using NUMA-aware reader-writer locks, we note that MCS locks appear to show do use MCS locks, which work well in NUMA systems.
% Other modifications explored include the replication of internal nodes, the ideal granularity of locking, and the performance of lock-free vs locking reads. The author did not implement or test the delete operation, but does sketch the changes necessary for it to be added. No implementation was provided.

\vspace{1mm}\mypara{Concurrent tries}
Tries are an alternative to B-trees for implementing concurrent (ordered) dictionaries.
The Masstree~\cite{Masstree} and the Adaptive Radix Tree with optimistic lock coupling (\textbf{OLC-ART})~\cite{ART, ART_OLC} both use optimistic concurrency control techniques.
In both, operations are accelerated using SIMD instructions.
However, they are not strictly comparison-based, and they require the programmer to serialize keys to be binary comparable.
This extra data marshalling is tedious and can add overhead.
Additionally, the shape and depth of the trees are determined by the key distribution, not by the number of keys they contain.\footnote{Height bounds in a trie are logarithmic in the size of the \textit{universe}. Even with path compression, some key distributions can result in deep tries.}
We compare with ART with optimistic lock coupling. ART with optimistic lock coupling has been shown to be faster than Masstree~\cite{ART_OLC}.

\vspace{1mm}\mypara{Distribution/contention aware data structures}
There has been also some work on data structures that are designed to accommodate non-uniform distributions. The concurrent interpolation search tree (\textbf{C-IST}) of Brown et~al.~\cite{C-IST} provides doubly-logarithmic runtime for smooth distributions.
However, its updates are slow.

The splay tree~\cite{SplayTree} is a popular sequential data structure that adapts to non-uniform distributions. After searching for a key, the splay tree performs rotations to move the node containing the key to the root. This reduces future access time for searches on the same key but also introduces a point of contention at the root, which makes the splay tree unsuitable for concurrent use. The \textbf{CBTree}~\cite{CBTree} is a concurrent splay tree-like data structure which uses counting to perform splaying only after a significant number of searches/updates have accessed a node, effectively amortizing the cost of the splay over many operations.

The \textbf{Splay-List}~\cite{Splay-List} is a concurrent variant of a SkipList~\cite{SkipList} that splays by increasing the height of frequently-accessed keys. Like the CBTree, it uses a counter-based approach to amortize the cost of the splaying.

% To the best of our knowledge, there is only one published data structure that adapts to different contention levels.
The contention adapting search tree (\textbf{CATree})~\cite{CATree} is a variant of an external search tree with binary internal nodes.
Each external node is a sequential dictionary data structure, protected by a lock.
AVL trees were used as the sequential dictionary in the authors' experiments, as well as our own.
The authors approximate contention at each external node by measuring how often a lock is already acquired when a thread attempts to acquire it.
When sufficient contention is detected at a node, the sequential data structure is split into two and an internal node connecting them is linked into the tree.
Similarly, two adjacent sequential data structures are combined if neither is under contention.

% While it would be interesting to compare our algorithms with the CATree, it is only implemented in Java, and is rather complex. % (over 2000 lines of code).
% We make a brief qualitative comparison.
% Consider the behaviour of the CATree in a workload where a single \textit{hot} key is frequently updated.
% These updates will cause the external node containing the hot key to split repeatedly until it contains a single key.
% It is likely that updates to this external node will be less efficient than updates to a node in a typical binary search tree (because of the additional overhead needed to maintain the sequential data structure).

%Thus, every access to that key will 
%then the performance of updates on this key will be similar to a typical binary search tree (but likely with higher overhead).
%While this reduces cache invalidations for keys situated near a hot key, performance on the hot key will still degrade to sequential performance.
%The algorithm also chooses between using lock-based and lock-free sequential data structures based on the contention level. This improves search performance in high-contention scenarios while avoiding the overhead of lock-freedom in low-contention scenarios.

\vspace{1mm}\mypara{General approaches}
There are several universal constructions for transforming sequential data structures into concurrent ones. These come in lock-based, lock-free, wait-free, and even NUMA-aware variants~\cite{PSIM, RedBlue, NUMA-UC}. Though they are simple to use, these constructions either require a copy of the data structure per thread or NUMA node (which is not practical for large data structures) or have a single global bottleneck on updates (e.g. an update log or state object). % which makes them scale poorly with updates.

Transactional memory makes it relatively easy to produce concurrent implementations of data structures, but it has significant drawbacks. Hardware transactional memory (HTM) is not universally available, and software transactional memory (STM) adds substantial overhead.
Furthermore, transactional memory is optimized for low-contention workloads.
In the high-contention scenarios we study in this paper, almost all transactions would abort (or serialize) because of data conflicts.
We performed some formative experiments comparing our trees with analogous trees implemented using HTM, STM, and a hybrid of the two (HyTM,~\cite{HyTM}), and found that, while the fastest of these implementations was close in performance to our trees under very low contention, performance degraded drastically under high contention.
We omitted these experiments, as they are only tangentially related to this work. %not sufficiently in scope. %out of scope. % showed that the performance of binary search trees that use transactional memory was on par with our trees on low-contention workloads but did not scale upon increasing contention.

\vspace{1mm}\mypara{Elimination}
Elimination was first introduced for use in concurrent stacks by Shavit and Touitou in~\cite{EliminationTrees}, but this implementation was not linearizable.
% Elimination in a stack reduces contention by trying to pair up concurrent push and pop operations, allowing threads to communicate with one another directly to complete their operations without modifying the stack.
The first linearizable implementation of elimination in stacks was provided by Hendler et al.~\cite{StackElimination}.
% , and improved performance in high contention workloads by up to 2.3x.
Hendler et al. coordinate the threads using an elimination array that stores ongoing operations' descriptors. 
Without loss of generality, suppose a thread $t$ is performing a push. $t$ first attempts to modify the data structure directly. If it fails, $t$ selects a random slot in the elimination array. If this slot contains a descriptor for a pop, $t$ attempts to eliminate both operations. Otherwise, if the slot is empty, it writes its own descriptor and waits a set amount of time to be eliminated.
% , before retrying its operation on the stack.
Note that it is possible for multiple push-pop pairs to be eliminated at once (at different indices in the elimination array). This is key to the scalability of the algorithm.
Braginsky et al. applied a similar approach to priority queues~\cite{CBPQ}.

\vspace{1mm}\mypara{Combining}
A different approach to tackling high contention workload is combining, in which a \textit{combiner} thread aggregates and performs the operations of many concurrent threads on the data structure.
% There are several ways the combiner thread can synchronize with other threads.
Drachsler-Cohen and Petrank provide an insightful summary of combining techniques~\cite{LocalCombiningOnDemand}.
Flat combining~\cite{FlatCombining} is one of the most popular techniques.
In flat combining, each thread attempting to update the data structure adds a record of its operation to a global list.
Threads compete to become the combiner by acquiring a global lock.
The combiner scans the entire list of operations, then performs them in some order.

Flat combining introduces higher latency compared to our publishing elimination technique.
Threads must \textit{wait} for the combiner to complete their operations one-by-one, and the wait can be quite long for operations near the end of the list. 
% Furthermore, 
% If the flat combining implementation reorders operations to remove/combine redundant operations, the combiner must complete its scan before performing any operation.
% Otherwise, if the flat combining implementation does not reorder operations
% Note that the combiner cannot begin to perform operations until the scan is complete, since later operations might 
% In contrast, elimination algorithms (including our publishing elimination) can eliminate multiple threads in parallel.

Recently, Drachsler-Cohen and Petrank created a variant of flat combining called local combining on-demand and demonstrated it on a linked list~\cite{LocalCombiningOnDemand}. They perform flat combining at each node in the list.
We tested our trees with a similar technique: We augmented each leaf node with an MCS queue~\cite{MCS} and used the queues to perform flat combining.
We found that this approach was much slower than our publishing elimination technique, in which threads do not have to wait for a combiner.

\vspace{1mm}\mypara{Persistent concurrent trees}
Venkataraman et al. introduced the CDDS-tree, a persistent concurrent B-tree. However, the pseudocode contains a global version number which is a scalability bottleneck~\cite{CDDSTree}.
Yang et al. created the NV-Tree, a persistent B-link tree that outperforms the CDDS B-tree by up to 12x~\cite{NVTree}.
The NV-Tree rebuilds \textit{all} of its internal nodes if \textit{any} internal node becomes too full. 
This can be extremely slow for large trees but occurs less than 1\% of the time in their workloads.
Additionally, the NV-Tree only persists leaf nodes since the entire tree can be recovered from them after a crash. This makes the recovery procedure slow, but avoids some flushes during normal execution.
% Unfortunately, it appears that the NV-Tree is not durably linearizable. When an insert into a full leaf node results in a split, a search may find the inserted key before the key is persisted. If a crash occurs before the pointer to the new node is persisted, the recovered data structure will not contain the inserted key. In this case, a search that sees the inserted key and returns before the crash cannot be linearized.

The \textbf{FPTree} is another persistent concurrent B-tree~\cite{FPTree}. It includes a number of optimizations that make it scale better than the NV-Tree.
Each leaf node includes a one-byte hash of each of its keys, known as a fingerprint. The fingerprints are scanned prior to probing the keys themselves, which limits the average number of key comparisons to 1. This can have a large impact when key comparisons are costly (for example, if the keys are strings).
Like the NV-Tree, the FPTree uses unsorted leaves and only persists leaf nodes.

The \textbf{RNTree} is a persistent concurrent B+tree that uses transactional memory and an indirection array with pointers to key-value pairs in each leaf node~\cite{RNTree}.
The indirection array makes binary searching for a key possible, with the drawback that inserts might require shifting every key-value pointer in the indirection array.

% The FPTree uses a combination of transactional memory (for internal nodes) and fine-grained locks (for leaf nodes) to synchronize threads. Using transactional memory for internal nodes simplifies the correctness argument for searches.
% Though the authors though it would be ideal to also use transactional memory for updates, transactional memory cannot be used in conjunction with persistence primitives. Thus, the FPTree's leaves have locks that are acquired by updates (and read by searches during their transaction to avoid seeing incomplete updates).

% Unfortunately, this implementation appears to share the same fundamental validation errors as the volatile one.

Finally, there are a number of general transformations for making concurrent dictionaries persistent.

RECIPE~\cite{RECIPE} provides general advice on how to make three categories of data structures persistent: those whose updates occur atomically, those whose updates fix inconsistent state, and those whose updates do not fix inconsistent state.
The \occ\ is closest to the third category.
RECIPE offers only a vague idea of how one might transform such a data structure.
In particular they instruct the data structure designer to: ``Add [a] mechanism to allow [updates] to detect permanent inconsistencies. Add [a] helper mechanism to allow [updates] to fix inconsistencies.''
Both of these seem to require deep knowledge of the data structure.
They also introduce fences after each store, whereas we carefully avoid fences where possible in the \occ.
%Converting the \occ\ to fit into one of these categories might hurt performance.

The transformations in NVTraverse~\cite{NVTraverse} and Mirror~\cite{Mirror} both provide durable linearizability, but target non-blocking data structures (and so are not applicable to the \occ). Montage~\cite{Montage} is another transformation which guarantees a weaker correctness condition known as buffered durable linearizability.

\section{\occ}\label{sec:occ_ABtree}
\mypara{Semantics}
The \occ\ implements the following \textit{dictionary} operations. %semantics for ordered set operations throughout this paper.
\begin{compactitem}
    \item \var{find(k)}: If a key-value pair with key \var{k} is present, return the associated value. Otherwise, return $\bot$.
    \item \var{insert(k, v)}: If a key-value pair with key \var{k} is present, return the associated value. Otherwise, insert the key-value pair \var{<k,v>} and return $\bot$.
    \item \var{delete(k)}: If a key-value pair with key \var{k} is present, delete it and return the associated value. Otherwise, return $\bot$.
\end{compactitem}
Range queries for the trees we present could be added using the techniques described in~\cite{EBRRangeQueries}.

% \subsection{Implementation}
The \occ\ consists of an \var{entry} pointer to a \textit{sentinel node} that is never removed.
This sentinel node has no keys and just one child pointer, which points to the root of the tree.
% The \occ\ is a tree of \var{Node}s represented by an \var{entry} node, which is never removed. The \var{entry} node has no keys and just one child: the root of the \occ.
The pseudocode for the data structures used in the \occ\ and selected operations are presented below.

\vspace{-2mm}
\subsection{Data structures}
\label{sec:structures}

The \occ\ has three types of nodes: leaf nodes, internal nodes and tagged internal nodes.
% All nodes contain between 2 and 11 pointers in our implementation (i.e. $a=2$ and $b=11$).
Leaf nodes store keys and values in their \var{keys} and \var{vals} arrays.
We say an entry in the \var{keys} array is \textbf{empty} if it is $\bot$.
An empty key has no associated value.
The keys in a leaf are \textit{unsorted} and there can be empty entries between keys.
This results in much faster updates since inserts and deletes do not need to shift other keys in the node.

Internal nodes contain $k$ child pointers (between 2 and 11, in our implementation), and $k-1$ \textit{routing} keys (that are used to guide searches to the appropriate leaf) in a \textit{sorted} array. 
Once an internal node is created, its routing keys are \textit{never} changed, but its child pointers can change.
To add or remove a key in an internal node, one must \textit{replace} the internal node.
This happens relatively infrequently.

A \tagged\ node (or simply \textbf{tagged node}) conceptually represents a temporary height imbalance in the tree.
A tagged node is created when a key/value must be inserted into a node but the node is full. The node is split, and the two halves are joined by a tagged node. Tagged nodes are not part of any other operation, and thus always have exactly two children.
% In a traditional B-tree, all leaves have the same depth.
% The \occ\ satisfies a \textit{relaxed} version of this balance property: all leaves have the same \textbf{tagged depth}.
% The tagged depth of a leaf is its depth, excluding any tagged internal ancestors.
Tagged nodes are eventually removed from the tree when the \var{fixTagged} rebalancing step is invoked.
% , which merges the tagged node into its parent.
% If the parent is full, this triggers another split and the creation of another tagged node higher up the tree (unless the root is being split, in which case the new node is just a regular internal node).

Each node has a \var{lock} field.
We use MCS locks as our lock implementation~\cite{MCS, CompactNUMAAwareLocks}. In MCS locks, threads waiting for the lock join a queue and spin on a \textit{local} bit (meaning they scale well across multiple NUMA nodes). In our trees, a thread only modifies a node if it holds its lock. Leaf nodes have an additional version field, \var{ver}, that records how many times the leaf has changed and whether it is currently being changed.
After acquiring a leaf's lock, a thread increments the version before making any changes to the leaf and increments the version again once it has completed its changes, and finally releases the lock.
Thus the version is even if the leaf is not being modified and odd if it is being modified.
The version is used by searches to determine whether any modifications occurred while reading the keys of a leaf\footnote{A leaf's version field could hypothetically wrap around and cause an ABA problem, but at 100 million updates per second, this would take 2900 years for a 64 bit word size.}.
% Separating the version from the lock also increases the window in which searches can read a consistent snapshot the leaf.
% The version is \textbf{locked} whenever the least significant bit is 1 and unlocked if it is 0.
% The other bits contain a sequence number (or version)

% The lock is acquired by incrementing an unlocked (even) value by one.
% To release the lock, the lock word is simply incremented. % to the next (even) number, thus incrementing the version and resetting the lowest order bit to 0 (unlocked).
%Searches use the versioned lock in the following way: the lock word is read once, then the leaf is searched, then the lock word is read again. If the lock was unlocked and had the same version number for both reads, either no update occurred between the reads of the lock word or the versioned lock was locked and unlocked so many times that it overflowed back to the same value (which, for a 32-bit lock, is extremely unlikely). Inserts and deletes that modify a leaf acquire the versioned lock before the modification and release it immediately after.
Nodes also contain a \var{marked} bit, which is set when a node is unlinked from the tree so that updates can easily tell whether a node is in the tree.
Marked nodes are never unmarked.

The \pathinfo\ structure is returned by \search\ and contains the node at which the search terminated, the node's parent and grandparent, the index of the node in the parent's \var{ptrs} array, and the index of the parent in the grandparent's \var{ptrs} array.

\begin{figure}
\centering
\begin{lstlisting}[linewidth=\columnwidth, numbers=left]
// K is key type, V is value type
abstract type Node
  keys      : K[MAX_SIZE]
  lock      : MCSLock @$\label{line:spinlock}$@
  size      : int
  marked    : bool

type Leaf inherits Node
  vals      : V[MAX_SIZE]
  ver       : int

type Internal inherits Node
  ptrs      : Node[MAX_SIZE]

type TaggedInternal inherits Internal

// The result of a search
type PathInfo
  gp     : Node // grandparent
  p      : Node // parent
  pIdx : int
  n      : Node // node
  nIdx : int

type RetCode is SUCCESS or FAILURE or RETRY

// Sentinel node: points to root
entry    : Internal
MIN_SIZE = 2, MAX_SIZE = 11
\end{lstlisting}
\vspace{-5mm}
\caption{\occ\ data structures}
\label{alg:datatypes}
\vspace{-5mm}
\end{figure}

\begin{figure}[t]
\centering
\begin{lstlisting}[linewidth=\columnwidth, numbers=left]
<RetCode, V> searchLeaf(leaf, key)
RETRY:
  ver1 = leaf.ver
  if ver1 is odd
    goto RETRY

  val = @$\bot$@
  for keyIndex = 0 up to MAX_SIZE - 1
    if leaf.keys[keyIndex] = key
      val = leaf.vals[keyIndex]
      break
  ver2 = leaf.ver @$\label{line:searchVer2Read}$@
  if ver1 @$\neq$@ ver2 goto RETRY
  if val = @$\bot$@ return <FAILURE, @$\bot$@>
  else return <SUCCESS, val>

PathInfo search(key, targetNode)
  gp = NULL, p = NULL, pIdx = 0, n = entry, nIdx = 0
  while n is not Leaf
    if n = targetNode break
    gp = p, p = n, pIdx = nIdx, nIdx = 0
    while nIdx < node.size-1 and key @$\geq$@ node.keys[nIdx]
      nIdx++
    n = n.ptrs[nIdx]
  return PathInfo(gp, p, pIdx, n, nIdx)

V find(key)
  path = search(key, NULL)
  rc, val = searchLeaf(path.n, key)
  return val
\end{lstlisting}
\vspace{-5mm}
\caption{\occ\ search operations}
\label{alg:search}
\vspace{-4mm}
\end{figure}

% \newpage
\subsection{Operations}
\begin{figure*}[t]
\vspace{-3mm}
    \centering
    \includegraphics[width=\textwidth]{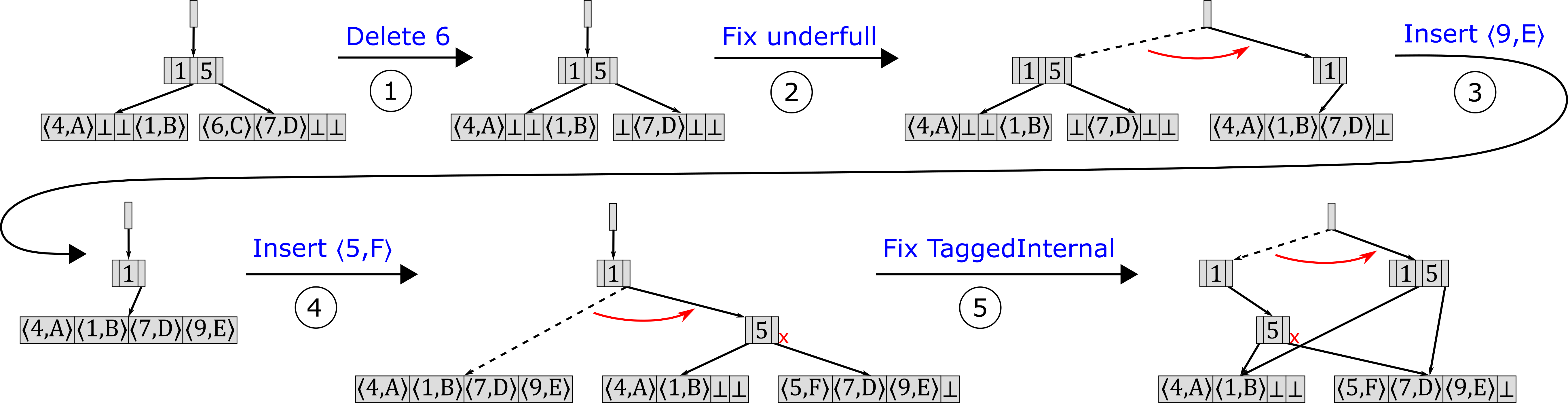}
    \vspace{-6mm}
    \caption{An \occ\ with $a = 2, b = 4$. (1)~The key-value pair $\langle 6,$C$\rangle$ is deleted. This creates an underfull node. (2)~The underfull node is merged with its sibling. This leaves the parent underfull, but the parent is the root, which is allowed to remain underfull. (3)~$\langle 9,$E$\rangle$ is inserted into an empty slot (\textit{simple} insert). (4)~No empty slot exists for $\langle 5,$F$\rangle$, so the appropriate leaf is split and a \var{TaggedInternal} node is created (\textit{splitting} insert). (5)~The \var{TaggedInternal} node is conceptually merged into its parent. We implement this by replacing it with a \textit{new} \var{Internal} node.} % that contains the keys (which  the \var{TaggedInternal} merged into its parent) and the grandparent's pointer is (atomically) written.}
    \label{fig:ABtreeOps}
% \vspace{-3mm}
\end{figure*}

All operations invoke a common \search\ procedure, which takes a \key\ and optionally a target node as its arguments, and searches the tree, starting at the root, looking for \key.
% Searches in the \occ\ start at the root.
At each internal node, \search\ determines which child pointer it should follow by traversing the (sorted) routing keys sequentially.
Once \search\ reaches a leaf (or the target node), it returns a \var{PathInfo} object as described in Section~\ref{sec:structures}.

\var{searchLeaf} is similar to the classical double-collect snapshot algorithm~\cite{DoubleCollect}.
It reads the leaf's version, reads its keys and values, then re-reads the leaf's version to verify that the leaf did not change while the keys and values were being read.
If the leaf \textit{did} change, then \var{searchLeaf} retries.
If the key is found, \searchleaf\ returns \var{<SUCCESS, val>}, otherwise, it returns \var{<FAILURE, $\bot$>}. Note that \search\ and \var{searchLeaf} do not acquire locks.
This allows for greater concurrency since internal nodes can be updated \textit{while} searches are traversing through them.

The \findargs\ operation simply invokes \search\ and \searchleaf, and returns \val.
\find\ operations in the \occ\ never have to restart, unlike in other trees.

\begin{figure}
\centering
\begin{lstlisting}
V insert(key, val)
RETRY:
  path = search(key, NULL)
  rc, val = searchLeaf(path.n, key)
  if rc = SUCCESS return val
    
  leaf, parent = path.n, path.p

  Lock leaf
  if leaf.marked @$\label{line:insertLeafMarked}$@
    Unlock leaf and goto RETRY

  // Verify key is not present
  for i = 0 to DEGREE - 1
    if leaf.keys[i] = key
      Unlock leaf and return leaf.vals[i]

  if leaf.size < MAX_NODE_SIZE
    // Insert without splitting
    for i = 0 to DEGREE - 1
      if leaf.keys[i] = @$\bot$@
        leaf.ver++ // Start modification
        leaf.keys[i] = key
        leaf.vals[i] = val
        leaf.size++
        leaf.ver++ // End modification@$\label{line:insSecondInc}$@
        Unlock leaf and return @$\bot$@
  else
    Lock parent
    if parent.marked
      Unlock leaf and parent and goto RETRY

    N = {contents of leaf} @$\cup$@ {key/val}
    newLeaf = TaggedInternal with two children that evenly share N
    parent.ptrs[path.nIdx] = newLeaf@$\label{line:insChangeParent}$@
    node.marked = true
    Unlock leaf and parent
    fixTagged(newLeaf)
    return @$\bot$@
\end{lstlisting}
\vspace{-5mm}
\caption{\occ\ insert operation}
\label{alg:insert}
\vspace{-4mm}
\end{figure}

\begin{figure}
\centering
\begin{lstlisting}
V delete(key)
RETRY:
  path = search(key, NULL)
  rc, val = searchLeaf(path.n, key)
  if rc = FAILURE
    return @$\bot$@

  leaf = path.n
  Lock leaf
  if leaf is marked
    goto RETRY

  for i = 0 to DEGREE - 1
    if leaf.keys[i] = key
      deletedVal = leaf.vals[i]
      leaf.ver++       // Start modification
      leaf.keys[i] = @$\bot$@
      leaf.size--
      leaf.ver++       // End modification

      if leaf.size < MIN_NODE_SIZE
        Unlock leaf
        fixUnderfull(leaf)
  return @$\bot$@
\end{lstlisting}
\vspace{-5mm}
\caption{\occ\ delete operation}
\label{alg:delete}
\vspace{-4mm}
\end{figure}

\vspace{1mm}\mypara{Delete}
The update operations are perhaps best understood with reference to Figure~\ref{fig:ABtreeOps}.
In a \delargs\ operation, a thread first invokes \searchargs\ and \searchleaf.
If it does not find \key, then \del\ returns $\bot$.
Otherwise, it locks the leaf and deletes the key by setting it to $\bot$, and returns the associated value (Figure~\ref{fig:ABtreeOps}(1)).
If \key\ was deleted by another thread between \search\ and acquiring the lock, \del\ returns $\bot$.
If deleting the key makes the node smaller than the minimum size $a$, \var{delete} invokes \fixunderfull\ to remove the underfull node by merging it with a sibling (Figure~\ref{fig:ABtreeOps}(2)).

\vspace{1mm}\mypara{Insert}
In an \insargs\ operation, a thread first invokes \searchargs\ and \searchleaf.
If it finds the \key, then \ins\ returns the associated value (Figure~\ref{fig:ABtreeOps}(3)).
Otherwise, it locks the leaf and tries to insert \key\ (resp., \val) into an empty slot in the \var{keys} (resp., \var{vals}) array. We call this case a \textbf{simple insert}.
If there is no empty slot, \ins\ locks the leaf's parent and \textit{replaces} the pointer to the leaf with a pointer to a new tagged node whose two (newly-created) children contain the leaf's old contents and the inserted key-value pair (Figure~\ref{fig:ABtreeOps}(4)).
We call this case a \textbf{splitting insert}.
The pointer change, and hence the insert of \key, is atomic.
The insert then invokes \fixtagged\ to get remove the tagged node (Figure~\ref{fig:ABtreeOps}(5)).

\vspace{1mm}\mypara{Rebalancing}
\fixtagged\ attempts to remove a tagged node.
It first searches for the tagged node, returning if it is unable to find it. (This case only occurs if another thread has already removed the tagged node)
If \fixtagged\ finds the target node, it tries to get rid of it by creating a copy $c$ of its parent, with the tagged node's key and children merged into $c$, and changing the grandparent to point to $c$ (Figure~\ref{fig:ABtreeOps}(5)). However, if the merged node would be larger than the maximum allowed size, \fixtagged\ instead creates a new node $p$ with two new children, which evenly share the contents of the old tagged node and its parent (Figure~\ref{fig:fixTaggedSplit}). The grandparent is then changed to point to $p$.
($p$ is a tagged node, unless it is the new root, in which case it is simply an internal node).

Now we turn to \fixunderfull. \fixunderfull\ fixes a node $n$ which is smaller than the minimum size, unless $n$ is the root/entry node.
It does this by either distributing keys evenly between $n$ and its sibling $s$ if doing so does not make one of the new nodes underfull (Figure~\ref{fig:fixUnderfull}).
Otherwise, \fixunderfull\ merges $n$ with $s$  (Figure~\ref{fig:ABtreeOps}(2)). In this case, the merged node might still be underfull or the parent node might be underfull (if it was of the minimum size before merging its children).
Thus, \fixunderfull\ is called on the merged node and its parent.
\fixunderfull\ requires that $n$ is underfull, its parent $p$ is not underfull, and none of $n$, $p$, and $s$ are tagged.
If these conditions are not satisfied, \fixunderfull\ retries its search.

\begin{figure}[t]
    \centering
    \includegraphics{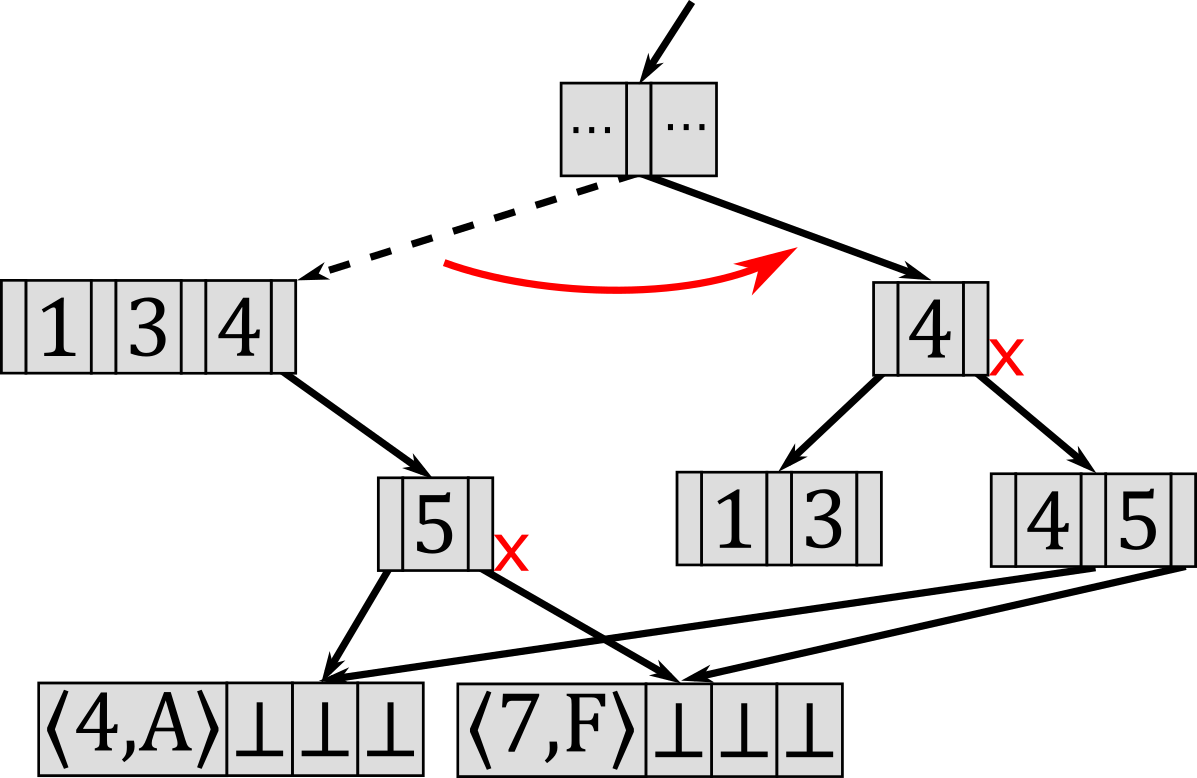}
    \vspace{-3mm}
    \caption{\fixtagged\ \textit{split} case (\textit{merge} is in Figure~\ref{fig:ABtreeOps})}
    \label{fig:fixTaggedSplit}
    \vspace{-3mm}
\end{figure}

\begin{figure}[t]
\centering
\begin{lstlisting}
fixTagged(node)
RETRY:
  if node.marked return
  path = search(node.searchKey, node)
  if path.n @$\neq$@ node return

  Lock path.n, path.p, and path.gp
  if node, parent or gParent is marked or
     path.p is TaggedInternal
    Release all locks and goto RETRY

  node.marked = true
  path.p.marked = true
  if path.p.size + 1 @$\leq$@ MAX_NODE_SIZE
    newNode = new Internal containing the keys & pointers of node and parent
    path.gp.ptrs[path.pIdx] = newNode
    Release all locks
  else
    // newNode is a TaggedInternal, unless it will be the new root (in which case it is Internal)
    newNode = new subtree of three nodes consisting of a new Internal that points to two new internal nodes which evenly share the keys & pointers of node and parent (except for the pointer to node)
    path.gp.ptrs[path.pIdx] = newNode
    Release all locks
      fixTagged(newNode)
\end{lstlisting}
\vspace{-5mm}
\caption{\occ\ fixTagged rebalancing step}
\label{alg:fixTagged}
% \vspace{-2mm}
\end{figure}

\vspace{-2mm} \subsection{Correctness}
\label{sec:correctness}
\vspace{-1mm}
This section proves that the \occ\ is linearizable. Recall that an algorithm is linearizable if, in every concurrent execution, every operation appears to happen \textit{atomically} at some point between its invocation and its response.

Proving the linearizability of the \occ\ requires a definition linking the \textit{physical} representation of the \occ\ (i.e. the contents of the system's memory) to the \textit{abstract} dictionary it represents. The operations are then shown to modify the physical state of the tree in a way that is consistent with the abstract semantics described at the beginning of Section~\ref{sec:occ_ABtree}.

\subsubsection{Definitions}
\begin{definition}[Reachable node]
\label{def:nodeInAbtree}
A node is said to be \textbf{reachable} if it can be reached by following child pointers from the entry node.
\end{definition}

\begin{definition}[Key in \occ]
\label{def:keyInAbtree}
Let $l$ be a reachable leaf. $k$ is \textbf{in the \occ} if, when $l$'s version was last even, $k$ was in $l$'s keys array.
Furthermore, if $k$ is the $i$th key in $l$, the value associated with $k$ is \var{l.vals[i]}.
\end{definition}

In other words, a key is logically inserted or deleted when a thread increments the version number of the leaf for the second time (making it even).

\begin{figure}[t]
    \centering
    \noindent\hspace{-3mm}\includegraphics{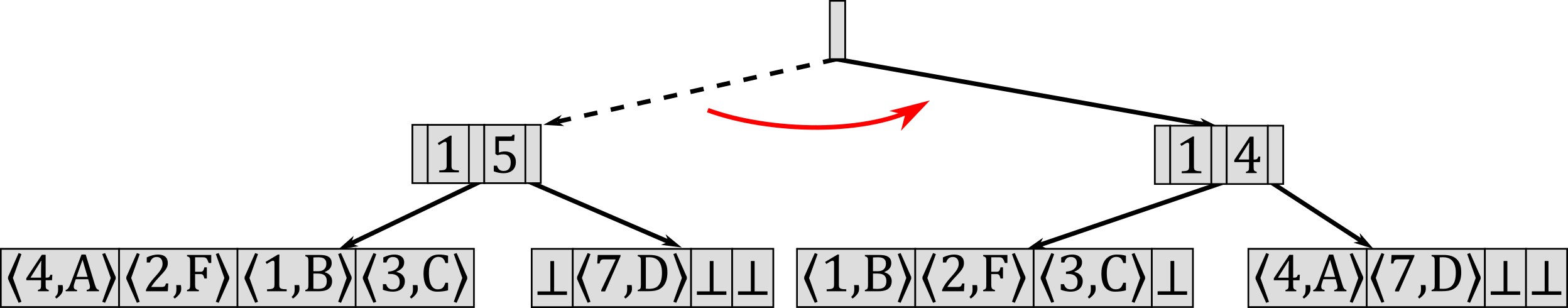}
    \vspace{-7mm}
    \caption{\fixunderfull\ \textit{distribute} case (\textit{merge} is in Figure~\ref{fig:ABtreeOps})}
    \label{fig:fixUnderfull}
    \vspace{-4mm}
\end{figure}

\begin{figure}[t]
\centering
\begin{lstlisting}
fixUnderfull(node)
  if node = entry or node = entry.ptrs[0] return

RETRY:
  path = search(node.searchKey, node)
  if path.n @$\neq$@ node return

  if path.nIdx = 0
    sIndex = 1 // Sibling is right child
  else
    sIndex = path.nIdx - 1 // 
  sibling = parent.ptrs[sIndex]

  Lock node, sibling, path.p, and path.gp
  if node.size @$\geq$@ MIN_NODE_SIZE return
  if parent.size < MIN_NODE_SIZE or
     node, sibling, parent, or gParent is marked or
     node, sibling or parent is TaggedInternal
    Release all locks and goto RETRY

  if node.size + sibling.size @$\leq$@ 2 * MIN_NODE_SIZE
    newNode, sibling = Distribute keys of node and sibling evenly amongst new node and sibling
    newParent = copy of parent plus pointer to newNode and key between newNode and sibling
    gParent.ptrs[path.pIdx] = newParent
      Mark node, parent, and sibling
      Release all locks and return
  else
    newNode = Combined keys of node and sibling
    if gParent = entry and parent.size = 2
      entry.ptrs[0] = newNode
      Mark node, parent, and sibling
      Release all locks and return
    else
      newParent = copy of parent with pointer to newNode instead of node/sibling
      path.gp.ptrs[path.pIdx] = newParent
      Mark node, parent, and sibling
      Release all locks
      fixUnderfull(newNode)
      fixUnderfull(newParent)
\end{lstlisting}
\vspace{-5mm}
\caption{\occ\ fixUnderfull rebalancing step}
\label{alg:fixUnderfull}
\vspace{-4mm}
\end{figure}

Definition~\ref{def:keyInAbtree} is somewhat counter-intuitive.
One might consider the following simpler definition: a key $k$ is in the tree if it is in some leaf's keys array.
Indeed, this alternate definition can also be used to prove that the \occ\ is linearizable. However, Definition~\ref{def:keyInAbtree} is necessary for the correctness of publishing elimination (Section~\ref{sec:elim}). Using a consistent definition hopefully makes the correctness argument easier for the reader.

There are two more definitions which are used in the proofs throughout this paper.
The \textbf{key range} of a node is a half-open subset of the universe of keys (e.g. $[100, 200)$ if the keys are numbers, or $[``aardvark", ``apple")$ if the keys are strings). Intuitively, the key range of a node is the set of keys that are allowed to appear in the subtree rooted at that node.
\begin{definition}[Key range]
\label{def:keyRange}
The key range of the entry node is defined to be the universe of keys. Let $n$ be a reachable internal node with key range $[L, R)$. If $n$ has no keys, the key range of its child is also $[L, R)$. Otherwise, suppose $n$ contains keys $k_1$ to $k_m$. The key range of $n$'s leftmost child (pointed to by \var{n.ptrs[0]}) is $[L, k_1)$, the key range of $n$'s rightmost child (pointed to by \var{n.ptrs[m]}) is $[k_m, R)$, and the key range of any middle child pointed to by \var{n.ptrs[i]} is $[k_i, k_{i+1})$.
\end{definition}

Finally, the \occ\ (along with all other trees in introduced in this paper) is a relaxed (a,b)-tree, as introduced by Larsen and Fagerberg~\cite{Larsen}. The relaxed (a,b)-tree is a search tree (as defined below). The most important consequence of the \occ\ being a search tree is that, for any key $k$ in the universe of keys, there is a unique search path for $k$, and this path passes through every reachable node whose key range contains $k$. Intuitively, this path is the path an atomic search of $k$ would take. Note that the uniqueness of the path implies that there is a unique reachable leaf in the search tree whose key range contains $k$.
\begin{definition}[Search Tree]
\label{def:searchTree}
Suppose $n$ is an internal node in a tree and $k$ is a key in $n$. A tree is a search tree if
\begin{itemize}
    \item All keys in the subtrees to the left of $k$ in $n$ are less than $k$ AND
    \item All keys in the subtrees to the right of $k$ in $n$ are greater than or equal to $k$
\end{itemize}
\end{definition}

\subsubsection{Invariants}
Proving an \ins\ is correct requires proving that \search\ finds the correct leaf to insert into. For \search\ to find the correct leaf, the tree must satisfy some structural properties, which are only satisfied if \textit{previous} inserts and deletes were correct.
We deal with this cyclical dependency is by assuming a set of invariants about the structure of the tree. These invariants hold for the initial state of the tree, and every modification to the tree preserves all invariants.
These invariants can then be used to prove the linearizability of the data structure.

\begin{theorem}[\occ\ Invariants]
The following invariants are true at every configuration in any execution of the \occ.
\begin{enumerate}
    \item The reachable nodes form a relaxed (a,b)-tree. \label{inv:searchTree}
    \item The key range of a node that was once reachable is constant. \label{inv:keyRangeConstant}
    
    \item A node that is not reachable contains the same keys and values that it contained when it was last reachable and unlocked (i.e. updates do not both unlink and modify a node). \label{inv:noModifyAndUnlink}
    \item A key appears at most once in a leaf. \label{inv:noDupKey}
    \item If a node was once reachable, and is currently unmarked, it is still reachable. \label{inv:linkedIfUnmarked}
    
    \item If a node is unlocked and was once reachable, its \var{size} field matches the number of keys it contains.
    \label{inv:sizeCorrect}
    \item The key range of \var{n} in \searchargs\ contains \key.
    
    \label{inv:searchCorrect}
\end{enumerate}
\end{theorem}

Intuitively, invariants \ref{inv:searchTree} to \ref{inv:noDupKey} follow from the sequential correctness of the updates together with the guarantee that any node that might be replaced or modified is locked and reachable until the update occurs. The sequential correctness of the updates  (i.e. their correctness in a single-threaded execution) can be established by inspection of the pseudocode, so we do not prove it in detail. We briefly explain the (concurrent) correctness of invariants \ref{inv:searchTree} to \ref{inv:noDupKey}. Invariants \ref{inv:linkedIfUnmarked} and \ref{inv:sizeCorrect} are straightforward from the pseudocode. 

Invariant \ref{inv:searchCorrect} is slightly different from the others, in that it is not a \textit{structural} invariant.
Rather, it describes the correctness of one of the operations.
The proof is somewhat involved, so it is proved in detail.

\begin{proof}
The invariants are satisfied at the initial state of the \occ.

\mypara{\ref{inv:searchTree}: \occ\ is a relaxed (a,b)-tree}
The updates to the tree are the same as those described by Larsen and Fagerberg in~\cite{Larsen}. They prove that, if these updates occur atomically, the tree is always a relaxed (a,b)-tree. Thus, the remainder of the proof simply shows that each update affects the tree atomically. This requires proving that for each update:
\begin{itemize}
    \item There is a single step at which the update appears to take place
    \item The update is correct
\end{itemize}
The first condition is simple. Simple inserts and deletes appear when the modified leaf is unlocked, by Definition~\ref{def:keyInAbtree}. All other updates only change a single pointer of a reachable node (to point to the update's newly created nodes).

For the second condition, assume that the updates are sequentially correct. This is easily verifiable by examining the pseudocode in this paper and comparing it to the pseudocode in~\cite{Larsen}.
To establish concurrent correctness from sequential correctness, it is sufficient to show that the update occurs on the correct data (i.e. on the correct node and with any preconditions of the sequential code satisfied), the update affects data that is actually in the tree, and the data used to construct the update does not change while the update is being constructed.

An \insargs\ operation uses the \search\ function to find the leaf in which to insert. By invariant~\ref{inv:searchCorrect}, the leaf's key range contains \key. By invariant~\ref{inv:searchTree} (this invariant), the tree is a relaxed (a,b)-tree and thus a search tree, so there is a unique \textit{reachable} leaf whose key range contains \key. Finally, this leaf is reachable if the insert returns, because \ins\ verifies that the leaf is not marked. Thus, the insert occurs in the correct leaf. A similar argument holds for \del.

The sequential code for the rebalancing steps has some preconditions. The \fixtagged\ rebalancing step requires that the node is tagged, but its parent and grandparent node are not.
\fixunderfull\ requires that none of the involved nodes are tagged, the parent node is not underfull, and the target node is underfull.
This is explicitly verified in both functions.
Thus, the rebalancing steps also act on the correct data.

Each update verifies that all involved nodes are not marked before performing its update. If the node is not marked, it is in the tree until the update itself unlinks the node, by invariant~\ref{inv:linkedIfUnmarked}. Moreover, any children of the node are also in the tree by Definition \ref{def:nodeInAbtree}. Thus, the data used to construct the update is actually in the tree.

Finally, the locks acquired by each update guarantee that any data involved in the update is constant until the locks are released.

\mypara{\ref{inv:keyRangeConstant}: Constant key range}
We must examine places where existing nodes are attached to a new parent and ensure that the key range of all descendent nodes remains the same. This happens in \fixtagged\ and \fixunderfull. In either function, the routing keys surrounding any pointer that is not removed remain the same before and after the update.
Thus, the key range of the pointed to node does not change.
This holds for leftmost and rightmost children of a node too, since the grandparent's key range does not change (by this invariant), and the new parent's key range is the same as the old parent's.

\mypara{\ref{inv:noModifyAndUnlink}: Unreachable nodes contain the same keys and values as they did when they were last reachable and unlocked}
Updates that unlink a node first lock it, then unlink it, then unlock it, \textit{without} changing the node's keys or values.

\mypara{\ref{inv:noDupKey}: No duplicate key}
Insert operations read the whole leaf while it is locked before attempting to insert a key, so a duplicate key is never inserted. The leaf that an insert operation tries to insert into is correct by invariant~\ref{inv:searchCorrect}.

\fixunderfull\ does not create duplicate keys when merging two leaves because there is a unique leaf whose key range contains a given key, and any keys in that key range are only present in that leaf (invariant~\ref{inv:searchTree}). Thus, a key can only be in one of the two leaves and so cannot appear twice in the merged node.

\mypara{\ref{inv:searchCorrect}: Search correctness}
The search maintains the invariant that the key range of the node it is currently reading contains the search key. Call this node $n$. The invariant is satisfied for the entry node, since its key range is the entire key space.
Since the routing keys of an internal node partition its key range, there is a unique child whose key range contains the search key.

Let $c$ be the child followed by the search after reading node $n$.
Even if $n$ is not in the tree when the pointer to $c$ is read by the search, $c$ must have been set as a child of $n$ while $n$ was in the tree, since only nodes in the tree are modified (invariant~\ref{inv:noModifyAndUnlink}). Thus, at the time that $n$ was in the tree and had $c$ as its child, the key range of $c$ contained the search key (Definition \ref{def:keyRange}). Since the key range of a node is constant (invariant \ref{inv:keyRangeConstant}), the key range of $c$ still contains the search key.
\end{proof}

With these invariants proved, the linearizability of the operations can now be established.

\subsubsection{Linearizability of \find}
The leaf at which \searchargs\ terminates was, at some point, the unique leaf that might have contained \key, by invariant~\ref{inv:searchCorrect}.

The search only returns if, during an interval when the leaf was unlocked, it either finds the search key and reads its value or it reads the entire leaf and does not find it. In the former case, we know that this key is unique in the leaf (invariant~\ref{inv:noDupKey}). Since the leaf was unlocked for the entire interval and nodes are not modified while they are unlocked, the result value of \find\ is correct for the leaf state in that interval.

If the leaf was in the tree at any point in this interval, \find\ may linearize at that point and be correct. If the leaf was never in the tree during the unlocked interval, \find\ linearizes at the point just before the leaf was unlinked. The leaf must have been marked when it was unlinked; so, by invariant~\ref{inv:noModifyAndUnlink}, the value returned by \find\ is the same as it would be if the \find\ occurred atomically just before the node was unlinked.

In this case, we must show that the \find\ was concurrent with the point when the leaf was unlinked. Theorem~\ref{the:noBackInTime} implies that the leaf must have been in the tree at some point during \find's invocation of \search. Since (by assumption) the node was not in the tree in the unlocked interval, the \find\ must have been concurrent with its unlinking. Note that the search procedure does not actually read the marked bit to check whether a leaf is in the tree; it is only described here for analysis.

\begin{theorem}
\label{the:noBackInTime}
Each node \search\ visits was in the tree at some time during the search.
\end{theorem}
\begin{proof}
The statement is true for the root. If the root is a leaf, the proof is complete. Otherwise, \search\ reads a child pointer from the root. We now show that any child pointer read from a node $n$ which was in the tree at some time during the search points to a child which was also in the tree at some time during the search.

If $n$ is still in the tree at the time the child pointer is read, the child pointed to is also in the tree at that point by Definition \ref{def:keyInAbtree}. Thus the child is also in the tree at some point during the search.

Otherwise, $n$ must have been (atomically) unlinked by some update $U$ at time $t$. The search was concurrent with the unlinking of $n$ since $n$ was in the tree at some point during the search (by assumption) and $n$ was not in the tree when the search read the child pointer.

By invariant \ref{inv:noModifyAndUnlink}, the pointers of $n$ point to its children just before it was unlinked at time $t$. Thus, the child followed by the search procedure was in the tree at some time during the search as well (namely, at $t$).
\end{proof}

\subsubsection{Linearizability of \ins\ and \del}
There are four possible linearization points for an \insargs\ operation. Note that in the final iteration of the \var{RETRY} loop, the leaf $l$ that the \ins\ locks is the \textit{unique} reachable leaf that might contain \key\ since the \occ\ is a search tree (invariant~\ref{inv:searchTree}), \key\ is in $l$'s key range (invariant~\ref{inv:searchCorrect}), and $l$ is not marked (invariant~\ref{inv:linkedIfUnmarked}).

An \ins\ that succeeds in its \search\ is linearized in the same way as a \find\ operation. The return value of the \search\ is the value associated with the key (by the correctness of \find) and is the correct value to return for \ins.

An \ins\ that finds \key\ in the leaf $l$ after acquiring the $l$'s lock (and thus does not modify $l$) may linearize at any point while the $l$'s lock is held because while $l$ is locked, the key cannot be removed from $l$, the key's associated value cannot change, and $l$ cannot be unlinked (since unlinking $l$ would require marking it). Since the leaf's version is even, the associated value is the correct return value according to Definition \ref{def:keyInAbtree}.

An \ins\ that inserts a key-value pair into a non-full leaf $l$ linearizes at its second increment of $l$'s version (which marks the modification as complete).
The key is not in the \occ\ \textit{before} the linearization point since the insert read $l$ while it was locked without finding the key, and $l$ is the unique reachable leaf that might contain $l$. The key \textit{is} in the \occ\ after the linearization point (according to Definition~\ref{def:keyInAbtree}) because the key is added to $l$, $l$ is still reachable, and $l$'s version is even.

For splitting inserts, searches can observe the change as soon as the pointer to the new subtree is written in the parent, since searches do not read locks on internal nodes. Thus, splitting inserts \textit{must} linearize at the write to the parent node.
Suppose a splitting insert writes the new pointer into the parent node $p$ at time $t$. Let $l$ be the leaf that was split and replaced by a tagged node $t$ with children $l_1$ and $l_2$. 
The inserted key is not in the \occ\ \textit{before} the write to $p$ since the insert reads $l$ while it is locked and does not find the key (and $l$ is the unique reachable leaf that might contain the key).
\textit{After} the write to $p$, the inserted key is in the tree because it is in either $l_1$ or $l_2$, both of which are reachable because $p$ is unmarked and thus reachable (invariant~\ref{inv:linkedIfUnmarked}).
The other keys in $l$ are not affected by splitting inserts since they are placed in one of $l_1$ or $l_2$ by the splitting insert.

The returned value of $\bot$ is correct in the above two cases, since the insert succeeded.
The linearization of deletes and justification of return values is similar to the first three cases above.

\subsubsection{Deadlock freedom}
Intuitively, deadlock freedom is guaranteed by locking order: nodes are locked from bottom to top, with ties broken by left-to-right ordering.
Note that the relative ordering never changes between two siblings, nor between parent and child. % never changes, nodes are never moved across a sibling or above their parent.

We have also created a version of the \occ\ with a height bounded by $O(\log(n) + c)$ height, where $c$ is the number of threads currently executing an operation on the tree. However, this version is slightly slower and has more complicated rebalancing logic.

\section{Elimination}\label{sec:elim}
We now describe a technique for eliminating \textit{dictionary operations} by carefully choosing the linearization order for concurrent insertions and deletions of the \textit{same key}. %so that only one operation $O$ modifies the data structure.
In the following, we say an insertion or deletion of \key\ is \textit{in progress} after it is invoked and before it returns.

Suppose $O$ is a simple \insargs.
If a deletion of \key\ is in progress when $O$ is linearized, then the delete can be linearized immediately before $O$ and return $\bot$ (without modifying the data structure).
Similarly, if an insertion of \key\ is in progress when $O$ is linearized, then the insert can be linearized immediately after $O$ and return \val.
Since neither of these operations change the data structure (when linearized in this way), an arbitrary number of insertions and deletions of \key\ can be eliminated, provided they are in progress when $O$ is linearized.

The case where $O$ is a successful \delargs\ is similar.
A deletion of \key\ that is in progress when $O$ is linearized can be linearized after $O$ (and return $\bot$), and an insertion of \key\ that is in progress when $O$ is linearized can be linearized before $O$ (and return the value removed by $O$).

\vspace{-1mm}
\subsection{Publishing elimination algorithm}
\vspace{-1mm}
The challenge is now to \textit{detect} insertions and deletions of \key\ that are in progress when $O$ is linearized.
We describe a modified version of the \occ\ called the \elim, in which each leaf additionally stores a summary, called an \elimrec, of the last operation $O$ that \textit{modified} it.
An \elimrec\ contains the following three fields. \key\ (resp. \val) stores the key (resp. value) that $O$ inserted or deleted.
\var{ver} stores a version number that helps an insert or delete determine whether it was in progress when $O$ is linearized.

Concurrent operations use the \elimrec\ to eliminate themselves as follows.
% ($O$ is either a simple \ins\ or a successful \del.)
Recall how a simple insert or successful delete $O$ modifies a leaf $l$.
It first increments the version number of $l$ to an odd value $v$, then modifies $l$, then increments $l$'s version number to the even value $v+1$.
It linearizes at this second increment.
$O$ publishes an \elimrec\ \var{rec} in $l$ just after the first increment.
\var{rec.ver} is set to $v$.\footnote{For simplicity, we only eliminate simple inserts and successful deletes. Eliminating splitting inserts would be more complicated and they are not as frequent.}
%
%%%%%%%%%%%%%%%%%%%%%%%%%%%%%%%%%%%%%

Observe that an insert or delete $O'$ is in progress when $O$ is linearized if the following conditions hold:
\begin{enumerate}[label=\textbf{C\arabic*.}]
    \item $O'$ reads $l.ver$ and sees it is $\le rec.ver$, and % less than or equal to \var{rec.ver} and
    \item $O'$ returns after $l.ver > rec.ver$ %It returns after the version of $l$ is at least \var{rec.ver+1}
\end{enumerate}

% We now show that both conditions are satisfied if an insert eliminates itself.
% To maximize the likelihood of elimination, an insert attempts to satisfy the first of these conditions by reading the version of the leaf as early as possible.
Let us see how an \insargs\ decides whether it can eliminate itself.
The insert first searches towards a leaf.
Once it arrives at a leaf $l$, it optimistically scans $l$ \textit{once} looking for \key.
(In contrast, in the \occ, \var{searchLeaf} is used to repeatedly scan $l$ until it obtains a consistent snapshot of $l$'s contents.)
%
% Suppose this single scan is consistent (i.e., \var{l.ver} is even and does not change during the scan).
% If the scan found \key, then no insert is necessary.
% Otherwise, 

If this single scan is not consistent, then the insert is concurrent with another update, so we try to eliminate it by invoking \lockOrElim\ (Figure~\ref{alg:elimination}).
%
% In the \occ, it would retry in However, it only tries to read a consistent value once in \var{searchLeaf} (instead of repeatedly retrying).
%If it fails to do so,
%then the insert invokes \lockOrElim. %\ (instead of \var{lock}).
\lockOrElim\ either eliminates the insert and returns \var{<false, rec.val>} (where \var{rec.val} is the value that the insert should return), or acquires the leaf's lock and returns \var{<true, $\bot$>}.
In the latter case, the insert then inserts \var{<key, val>} into $l$ and releases the lock (as in the \occ).

On the other hand, suppose the scan \textit{was} consistent.
If it found \key, then no modification is necessary, and the insert returns.
Otherwise, it will use \lockOrElim\ to try to lock $l$ so it can insert \key.
(If the insert experiences contention while acquiring the lock, it might even be eliminated.)
% Note that the \elim\ uses a test-and-test-and-set spinlock instead of an MCS lock, since it is difficult for threads to abort their attempt to acquire an MCS lock.
% (Pseudocode of this modified search procedure appears in Appendix~\ref{sec:trySearch}).

\begin{figure}
\centering
\begin{lstlisting}[linewidth=\columnwidth, numbers=left]
// K is key type, V is value type
type ElimRecord {key: K, val: V, ver: int}
type Leaf
  ...
  rec: ElimRecord

V insert(key, val)
  ... // Find leaf and search it once
  acq, retval = lockOrElim(leaf, key)
  if not acq
    return retval@$\label{line:retEliminated}$@

  // Did not eliminate, insert as usual
  leaf.ver++
  leaf.rec = <key, val, leaf.ver>@$\label{line:publishRec}$@
  ... // Insert key
  leaf.ver++
  Unlock leaf and return @$\bot$@
  ...

// Returns <true, _> if acquired
// Returns <false, val> if eliminated
<bool, V> lockOrElim(leaf, key)
  startVer = leaf.ver
  while true
    // Try to eliminate self
    do @$\label{line:lockOrElimVer1}$@
      ver1 = leaf.ver
      rec = leaf.rec
      ver2 = leaf.ver@$\label{line:secondRead}$@
    while ver1 is odd or ver1 @$\neq$@ ver2@$\label{line:lockOrElimVer2}$@

    if startVer @$\leq$@ rec.ver and rec.key = key@$\label{line:verCompare}$@
      return <false, rec.val>

    // Cannot eliminate, try to lock
    if leaf.lock.tryLock()@$\label{line:tryLock}$@
      return <true, _>
\end{lstlisting}
\vspace{-5mm}
\caption{Elimination pseudocode}
\vspace{-2mm}
\label{alg:elimination}
\end{figure}

\vspace{1mm}\mypara{How \var{lockOrElim} performs elimination}
%It remains to describe how \var{lockOrElim} actually eliminates
In \var{lockOrElim}, the insert attempts to read a snapshot of the leaf's \elimrec.
To do this, it reads the leaf's version (line~\ref{line:lockOrElimVer1}), then reads the \elimrec\ \var{rec}, then re-reads the leaf's version (line~\ref{line:lockOrElimVer2}).
If the reads of the leaf's version return identical results, and the version is even (indicating the leaf is not being modified), then a snapshot was obtained. %\var{rec} was not modified partway through its read. 
Otherwise, \var{lockOrElim} tries to obtain a snapshot again. %tries to read a snapshot of \var{rec} again.

Once a snapshot is obtained, condition C2 is guaranteed to be satisfied.
To see why, note that the leaf's version is even when it is last read at line~\ref{line:secondRead} by the exit condition of the loop.
But, \var{rec.ver} is always an odd value, thus the version read at line~\ref{line:secondRead} is at least \var{rec.ver+1}.

At line~\ref{line:verCompare}, \lockOrElim\ tries to determine whether condition 1 is satisfied.
If it is, and \key\ matches \var{rec.key}, then this insert can be eliminated. So, \lockOrElim\ returns \var{<false, rec.val>} and \ins\ returns \var{rec.val} at line~\ref{line:retEliminated}.
Otherwise, \lockOrElim\ %the insert does not have enough information to argue that it can eliminate itself, so it 
tries to lock the leaf at line~\ref{line:tryLock}.
% (The \elim\ uses test-and-test-and-set spinlocks, since MCS locks do not provide a \var{tryLock} function.)
If it acquires the lock, it returns \var{<true, $\bot$>}. %the insert proceeds as in the \occ.
% If the insert is a \textit{simple insert}, it publishes an \elimrec\ (as described above).
If \lockOrElim\ fails to acquire the lock, it attempts to eliminate the insert again.

The elimination of deletes is similar, except that eliminated deletes always return $\bot$ (not \var{rec.val}).
Figure~\ref{fig:pubelim} shows an example of publishing elimination.

\begin{figure}[t]
    \centering
    \includegraphics{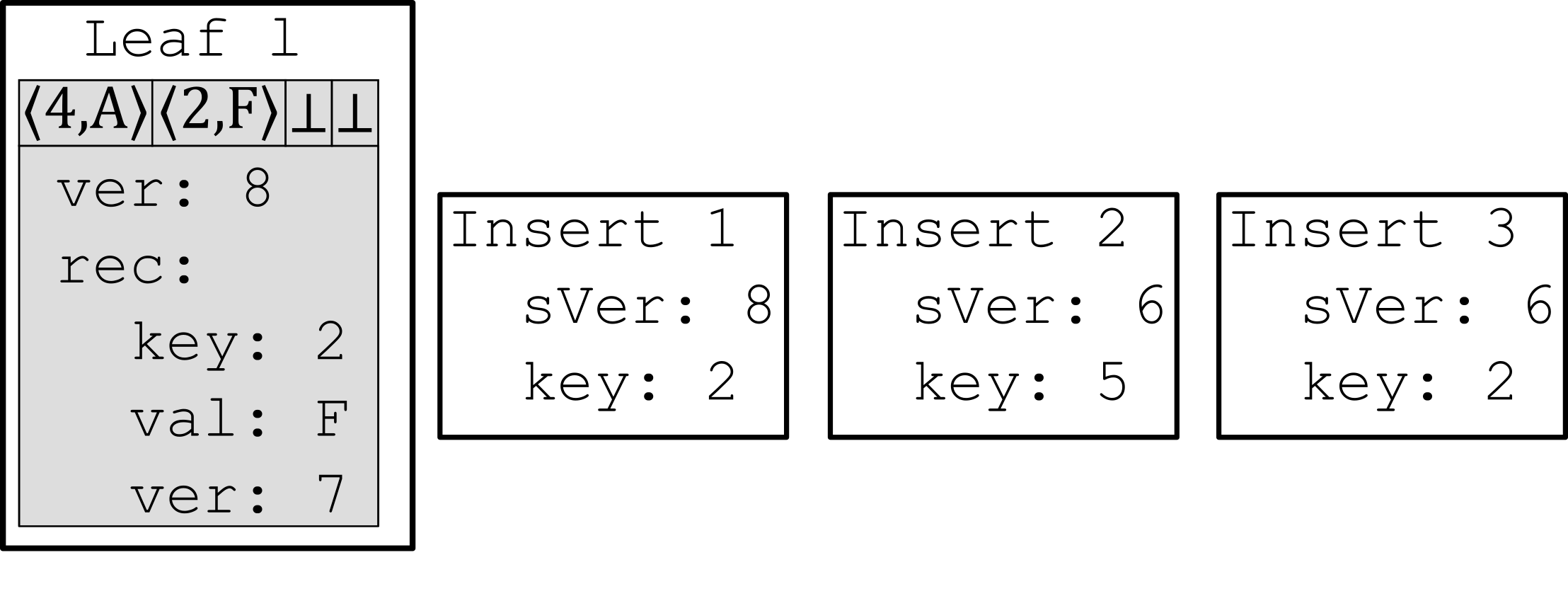}
    \vspace{-5mm}
    \caption{Consider the state of leaf $l$ as shown.
    $l.rec$ stores the \elim\ of a completed simple insert \var{\ins(2,F)}. Consider three (independent) inserts that are attempting to insert in $l$ and are all at line~\ref{line:verCompare}.
    Insert 1 cannot eliminate itself with \var{rec} since the version of the leaf it read is greater than \var{rec.ver}.
    Insert 2 cannot eliminate itself since its key does not match \var{rec.key}. Insert 3 can eliminate itself.}
    \label{fig:pubelim}
    \vspace{-2mm}
\end{figure}

% The \pubelim\ is based on the \occ\ instead of the \mcs\ since MCS locks do not work well with publishing elimination: for operations to be able to eliminate themselves, threads would have to be able to remove their own queue nodes from the MCS queue.
% Doing this would requires synchronizing with the neighboring queue nodes, which would make the algorithm complicated (and likely slower).

% In Algorithm~\ref{alg:PubElim}, the operation tries to eliminate itself before trying to acquire the lock since we are optimizing for the high contention case. In the low contention case, \var{lockOrElim} can be modified to optimistically try to acquire the spinlock first.
% If the read is not consistent, some operation must have concurrently modified the leaf. Since eliminating is preferable to acquiring the lock, the thread repeatedly tries to read the \var{rec} field until it can read a consistent version.

\mypara{Searches could be eliminated}
Finally, we note that the \elimrec\ could also be used to linearize \find s in high-contention workloads.
In some extreme scenarios, this could possibly be useful in preventing \findargs\ from being starved by an endless stream of updates to \key.
We did not observe this in our experiments, since our node size is small enough that searches can typically traverse a leaf in the interval between when one update completes and the next one begins.
%However, this could be an issue for larger node sizes.

\section{Persistent trees}
In this section we describe the changes to make a persistent version of the \occ, the \pmocc. The \pmocc\ persists only its keys, values, and pointers.
Every update in the \pmocc\ appears to occur atomically in persistent memory.
Thus, the recovery procedure for the \pmocc\ is extremely simple: it traverses the tree in persistent memory starting from the root (which is in a known location), and fixes all non-persisted fields (i.e. setting \var{size} to the actual number of pointers/values in the node, and resetting version, lock state, and the marked bit to their initial values).

The updates in the \pmocc\ require the following cache line flushes. (Below, a flush refers to a \var{clwb} instruction followed by an \var{sfence}).
For a simple \insargs, two flushes must be used: \val\ must be flushed after it is written, and \key\ must be flushed after it is written. The insert occurs atomically when the \key\ reaches persistent memory.
Note that if a crash occurs \textit{after} \val\ is flushed but before \key\ is, \key\ is still $\bot$ so the key-value pair is \textit{not} logically in the tree.
For a successful delete, \key\ must be flushed after it is set to $\bot$. The delete occurs atomically when the key field is equal to $\bot$ in persistent memory.

Recall that splitting inserts and rebalancing steps occur atomically in \textit{volatile} memory by creating a set of new nodes and linking them into the tree by changing a single pointer.
We guarantee that these updates appear atomically in \textit{persistent} memory by flushing the new nodes before changing the pointer, then flushing the pointer.
The update occurs atomically when the new pointer is flushed.

Operations in the \pmocc\ must only follow \textit{persisted} pointers.
To see why, consider the following scenario: a splitting \ins\ inserts \key\ and \val, then a \findargs\ operation returns \val, then a crash occurs before the pointer to the new nodes is persisted.
In this case, the recovered tree will \textit{not} contain the inserted key-value pair, so the \find\ cannot be linearized.
To ensure that all operations only access persisted data, we use the link-and-persist method from~\cite{LinkAndPersist} (a similar technique is given in~\cite{PMwCAS}).
In this technique, whenever an update writes a new pointer $p$ into the tree, $p$ is written with a mark on it to indicate that it has not been persisted.
It is then flushed, and the mark is removed.
Whenever a thread encounters a marked pointer, it waits until the mark is cleared (hence the pointer is flushed) before following the pointer.

There are two differences in the linearization points of the \pmocc.
First, splitting inserts must linearize when the new pointer is flushed. Operations cannot access the new key-value pair before this point because the pointer to the tagged node is still marked.
The second change is more subtle.
In the \occ\ and \elim, a simple insert or successful delete $O$ is linearized at the second increment of the version number. 
In the persistent setting, a crash could occur \textit{before} this increment but \textit{after} \key\ has been flushed, so the update will be recovered.
To deal with this, any simple insert or successful delete that flushes \key\ but has not yet incremented version for the second time when a crash occurs is linearized at the time of the crash.
These changes result in a durable linearizable implementation.

The \elim\ can also be made persistent by applying the same changes. We call the resulting tree the \pmelim. The change to the linearization point of $O$ does not affect the correctness argument for publishing elimination, since $O$ can only cause the elimination of another operation \textit{after} $O$ has incremented the version for the second time.

\subsection{\pmocc\ Correctness}
This section begins by providing a definition to link the physical state of the \pmocc\ to its abstract contents. It then mentions some invariants which hold for the \pmocc; these are analogous to the invariants of the \occ\ and can similarly be used to show that the \pmocc\ is strictly-linearizable.

\subsubsection{Definitions}
\begin{definition}[p-Reachable node]
\label{def:pReachable}
A node is \textbf{p-reachable} (short for persistently reachable) if it can be reached from the entry node by following child pointers in \textit{persistent} memory.
\end{definition}

\begin{definition}[Recovering]
\label{def:Recovering}
The system is said to be \textbf{recovering} from the time when a crash occurs until the time when the recovery procedure returns.
\end{definition}

In strict linearizability, every operation that is concurrent with a crash must either be linearized before the crash or be removed from the execution. Any simple insert or successful delete that has flushed a \key\ will be recovered (and thus cannot be removed from the execution). These operations must therefore be linearized before the crash, even if they have not yet incremented the version for the second time.
This is reflected in the definition below, and in the changes to the linearization points in the following section.
\begin{definition}
\label{def:keyInPmAbtree}
Let $l$ be a p-reachable node.
A key $k$ (not equal to $\bot$) is \textbf{in the \pmocc} if either
\begin{enumerate}
    \item The system is recovering and $k$ is in $l$'s keys array OR \label{def:keyInPmAbtree:recovering}
    \item The system is not recovering and $k$ was in $l$'s keys array when $l$'s version (in \textit{volatile} memory) was last even
    \label{def:keyInPmAbtree:normal}
\end{enumerate}
Furthermore, if key $k$ is the $i$th key in $l$, the value associated with $k$ is \var{l.vals[i]}.
\end{definition}

If the system is not recovering, Definition~\ref{def:keyInPmAbtree} is similar to the definition of a key being in the \occ. That is, keys and values are logically added or removed from the tree when the version number is incremented to an even number.
If the system is recovering, however, every key in a p-reachable node is in the tree (the version is ignored).

\subsubsection{Invariants}
\begin{theorem}[\pmocc\ Invariants]

The \pmocc\ satisfies the following invariants, which are analogous to the \occ\ invariants:
\begin{enumerate}
    \item The p-reachable nodes form a relaxed (a,b)-tree.
    \item The key range of a node that was once p-reachable is constant. 
    \item A node that is not p-reachable contains the same keys and values that it did when it was last p-reachable and unlocked (i.e. updates do not both unlink and modify a node).
    \item A key appears at most once in a leaf.
    \item If a node was once p-reachable, and is currently unmarked, it is p-reachable.

    \item If a node is unlocked and was once reachable, its \var{size} field matches the number of keys it contains.
    \item The key range of \var{n} in \searchargs\ contains \key.
\end{enumerate}
\end{theorem}
\begin{proof}
The proofs of most of these invariants are similar to the proofs in Section~\ref{sec:correctness}. The proof for invariant 1 requires an additional explanation of why every node used by an update was once p-reachable.

\mypara{Proof of invariant 1}
Recall that to prove the \pmocc\ is a relaxed (a,b)-tree, it suffices show that for each update:
\begin{itemize}
    \item There is a single step at which the update appears to take place
    \item The update is correct
\end{itemize}
The first condition holds for the reasons laid out in the previous section on atomic updates.

The second condition is largely the same as the proof in the \occ. However, that proof uses invariant~\ref{inv:linkedIfUnmarked}, which requires showing that the nodes traversed in the search were all reachable at some time. This was trivial in the case of the \occ\ (because the nodes are reached by following child pointers), but is not trivial in the \pmocc, which uses \textit{p-reachability}. 

We will show that every node traversed by \search\ in the \pmocc\ was p-reachable at some time.
Assume that every node traversed by a search until node $n$ is p-reachable.
If $n$ is the entry node, it is p-reachable by definition.

Otherwise, the search reached $n$ by following an unmarked pointer from a node $p$.
We will show that there exists a time $t$ when $p$ was p-reachable and contained the unmarked pointer to $n$.
If $p$ was p-reachable when it read the unmarked pointer to $n$, $t$ is the time of the read.
Otherwise, invariant~\ref{inv:noModifyAndUnlink} guarantees that $p$'s pointers have not been modified since it was last p-reachable. Thus, when $p$ was last p-reachable, it contained an unmarked pointer to $n$. In this case, $t$ is the time when $p$ was last reachable.

Finally, notice that there are two ways $p$ could contain an unmarked pointer to $n$.
The first way is if $p$ contained a pointer to $n$ when it was created. 
In this case, since $p$ was flushed before being linked into the tree, its pointer to $n$ is persisted.
Otherwise, the pointer was first introduced to $p$ by an update $U$ as a marked pointer. This update must have flushed the pointer before unmarking it.
In either case, the unmarked pointer was in persistent memory by time $t$.

At time $t$, $p$ was p-reachable and contained a pointer to $n$ in persistent memory.
Thus, $n$ was p-reachable at time $t$.

The remainder of the proof of is similar to the proof for the \occ.
\end{proof}

\subsubsection{Strict linearizability}
The \pmocc\ has slightly different linearization points than the \occ, to deal with the different definition of when a key is in the tree.

Operations which do not modify the tree are linearized as in the \occ.
Splitting inserts in the \pmocc\ are linearized when the pointer to the new nodes is flushed to persistent memory (instead of it is written to volatile memory).

Simple inserts and successful deletes linearize differently depending on whether or not they are interrupted by a crash.
When not interrupted by a crash, simple inserts and successful deletes have the same linearization points as they did in the \occ: the increment of the leaf's version (in volatile memory) to an even number.
Recall that this linearization point is chosen to support publishing elimination.

To see why we cannot linearize the same way when interrupted by a crash, consider the following scenario. 
Suppose a simple insert or successful delete that flushes \key\ (thus making its change persistent) but does increment the version to an even number before a crash.
The recovery procedure would recover this key-value pair, even though the operation was not linearized.

To solve this problem, we linearize these operations \textit{at the crash}.

\begin{theorem}
The \pmocc\ is strictly-linearizable.
\end{theorem}
\begin{proof}
Note that these linearization points all occur after an operation's invocation and before its response or crash. 
We must show that performing each operation at its linearization point (and returning the appropriate value) correctly affects the contents of the abstract dictionary according to Definition~\ref{def:keyInPmAbtree}.
Let $E$ be an arbitrary execution of the \pmocc.
We prove that $E$ is strictly-linearizable by induction.

Suppose the prefix of $E$ up to the beginning of the $i$th era of operations (including the recovery procedure after the $i-1$th crash, if $i > 1$) is strictly-linearizable, and that the \pmocc\ satisfies all invariants. 
We show that the prefix up to the beginning of the $i+1$th era of operations is strictly-linearizable and the \pmocc\ recovered after the $i$th crash satisfies all invariants.

We do this by breaking up the execution fragment from the beginning of the $i$th era of operations to the beginning of the $i+1$th era of operations into three parts: the execution fragment before the $i$th crash, the $i$th crash, and the recovery after the $i$th crash. 
We show that each fragment is strictly-linearizable by showing that the operations correctly modify the abstract dictionary.
Note that the concatenation of strictly-linearizable execution fragments is strictly-linearizable, by the \textit{locality} property of strict linearizability. 

\vspace{1mm}\mypara{Before the $i$th crash}
We consider the tree operations performed from the beginning of the $i$th era until (but not including) the $i$th crash.
The linearizability arguments for these operations is analogous to the arguments established for linearizability in the \occ:
the linearization points used in this section are all analogous to the \occ's linearization points, the \pmocc\ satisfies analogous invariants, and the definition of a key being in the \pmocc\ is Definition~\ref{def:keyInPmAbtree}.\ref{def:keyInPmAbtree:normal} (which is analogous to the \occ's definition of a key being in the tree).

\vspace{1mm}\mypara{At the $i$th crash}
At the time of the crash, the definition of a key being in the tree changes from Definition~\ref{def:keyInPmAbtree}.\ref{def:keyInPmAbtree:normal} to Definition~\ref{def:keyInPmAbtree}.\ref{def:keyInPmAbtree:recovering}.
It must be shown that the keys in the tree after the crash are exactly those that were in the tree before the crash, plus any that were inserted by a simple insert linearized at the crash and minus any that were deleted by a successful delete linearized at the crash.

First consider the case when a key $k$ \textit{is} in the tree after a crash. That is, there exists some p-reachable leaf $l$ such that \var{k = \texttt{l.keys[i]}} (for some index \var{i}).

If $k$ was also in the tree before the crash (according to Definition~\ref{def:keyInPmAbtree}.\ref{def:keyInPmAbtree:normal}), it is only correct for $k$ to be in the tree after the crash if no delete of $k$ linearized at the crash. This is indeed the case, since a delete of $k$ that linearized at the crash would have set \var{l.keys[i]} to $\bot$ and flushed $\bot$. But, by assumption, $k$ is in the keys array of $l$ after the crash.

Otherwise, if $k$ was not in the tree before the crash, it is only correct for $k$ to be in the tree after the crash if an insert of $k$ \textit{did} linearize at the crash. This is true.
Since $k$ was not in the tree before the crash but $l$ was p-reachable and contained $k$, the $l$'s version must have been odd at the crash (according to Definition~\ref{def:keyInPmAbtree}.\ref{def:keyInPmAbtree:normal}).
Thus, there must have been an ongoing insert that inserted $k$ at the time of the crash. Since the crash occurred after the flush of $k$ but before the version was incremented to an even number, this insert linearized at the crash.

A similar argument shows that $k$ is \textit{not} in the tree after a crash if and only if $k$ was either deleted at the crash or was not in the tree before the crash (and was not inserted at the crash).

\pmocc\ invariants~\ref{inv:searchTree}-\ref{inv:noDupKey} and \ref{inv:searchCorrect} are maintained during a crash since they only describe persisted data. Invariants~\ref{inv:linkedIfUnmarked} and \ref{inv:sizeCorrect} might be incorrect since they refer to volatile fields.
However, they are restored by the recovery procedure.

\vspace{1mm}\mypara{After the $i$th crash (recovery)}
The recovery procedure does not affect the set of p-reachable nodes or their keys or values, so the set of keys in the tree is fixed while the system is recovering.
By the time the recovery procedure returns, all p-reachable nodes' versions are 0, and thus the key in the tree is the same according to Definitions~\ref{def:keyInPmAbtree}.\ref{def:keyInPmAbtree:recovering} and \ref{def:keyInPmAbtree}.\ref{def:keyInPmAbtree:normal}.

Additionally, all \pmocc\ invariants are satisfied by the time the recovery procedure returns.
\pmocc\ invariants~\ref{inv:searchTree}-\ref{inv:noDupKey} and \ref{inv:searchCorrect} were correct before recovery, and the recovery procedure fixes the volatile fields, which ensures that invariants~\ref{inv:linkedIfUnmarked} and \ref{inv:sizeCorrect} hold by the time it returns.

Thus, the execution up to the beginning of the operation in the $i+1$th era is strictly-linearizable, and the \pmocc\ satisfies all invariants.
\end{proof}

The proofs for the \pmelim\ is similar.
Note that elimination does not conflict with the change of linearizing some operations at a crash.
In both the \pmocc\ and the \pmelim, an operation $O_e$ is only eliminated \textit{after} the successful operation $O_p$ has executed its second increment of the leaf's version. Any simple insert or successful delete has linearized by this time (and a future crash does not change this fact).

% Simple inserts are modified to flush the inserted value after it is written and the key after it is written.
% A simple insert is linearized when the key is written to persistent memory (somewhere between the write and the flush).
% The value must be flushed \textit{before} the key is written, since (in our implementation) keys and values are not necessarily in the same cache line, so the key might otherwise be unintentionally flushed before the value. This would lead to a key-value pair that was never inserted being in the tree. Successful deletes simply flush after they delete the key (which marks the key-value pair as invalid), and are linearized when the key is $\bot$ in persistent memory.

% Searches must also be modified to guarantee that they only reach persisted nodes. This is accomplished via the link-and-persist method from~\cite{LinkAndPersist}. In link-and-persist, when a pointer must be updated, the new value is marked to signal that it has not been flushed. The mark is only removed after the pointer has been flushed, and \search\ does not follow marked pointers.
% Thus, \search\ only reaches persisted nodes.
% Moreover, since updates lock the nodes they use and changes are guaranteed to be persisted before unlocking, updates also only use the data in persistent memory.

% Only keys, values, and pointers in the \occ\ are persisted. 
% The remaining fields can either be reconstructed (e.g. the \var{size} field) or reset to its initial value (e.g. version, lock state) by the recovery function.

\section{Experiments}\label{sec:Exp}
\begin{figure*}[t]
    \centering
    \setlength\tabcolsep{0pt}
    \vspace{-1mm}
\begin{minipage}{1\linewidth}
    \centering
    \begin{tabular}{m{0.04\linewidth}m{0.48\linewidth}m{0.48\linewidth}}
        &
        \fcolorbox{black!50}{black!20}{\parbox{\dimexpr \linewidth-2\fboxsep-2\fboxrule}{\centering {Zipf parameter = 0 (Uniform)}}} &
        \fcolorbox{black!50}{black!20}{\parbox{\dimexpr \linewidth-2\fboxsep-2\fboxrule}{\centering {Zipf parameter = 1 (Skewed)}}}
        \\
        \rotatebox{90}{100\% updates} &
        \includegraphics[width=\linewidth]{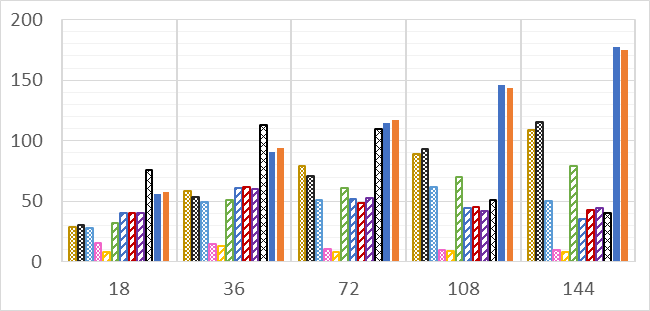} &
        \includegraphics[width=\linewidth]{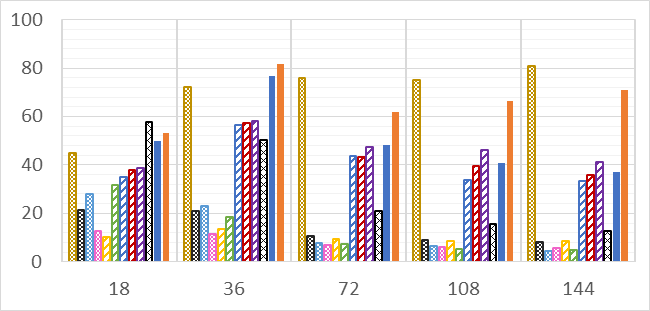}
        \\
        \vspace{-5mm}\rotatebox{90}{50\% updates} &
        \vspace{-6mm}\includegraphics[width=\linewidth]{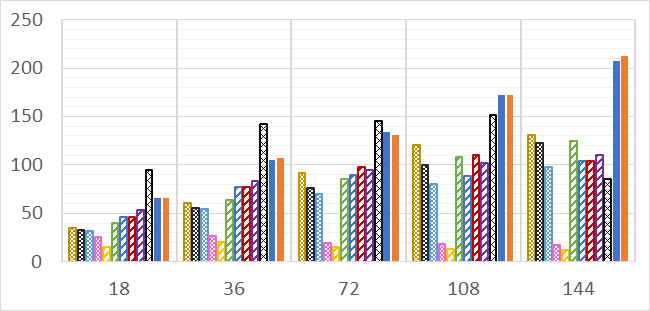} &
        \vspace{-6mm}\includegraphics[width=\linewidth]{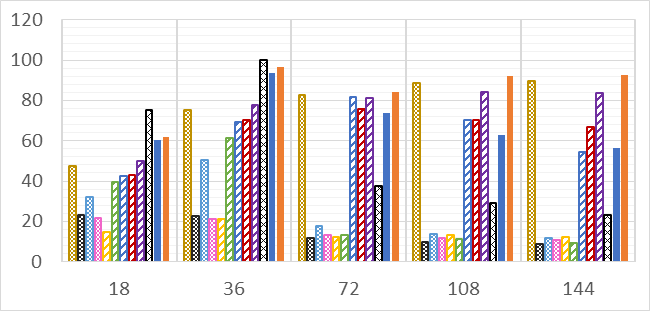}
        \\
        \vspace{-6mm}\rotatebox{90}{20\% updates} &
        \vspace{-6mm}\includegraphics[width=\linewidth]{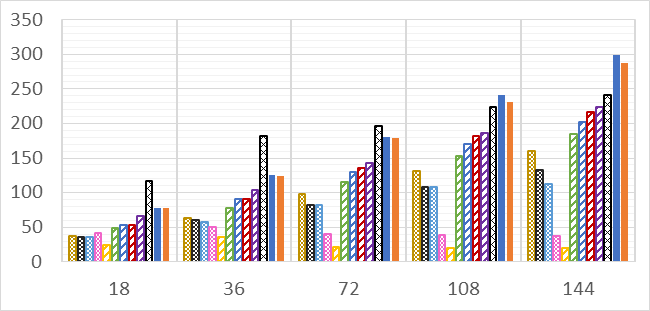} &
        \vspace{-6mm}\includegraphics[width=\linewidth]{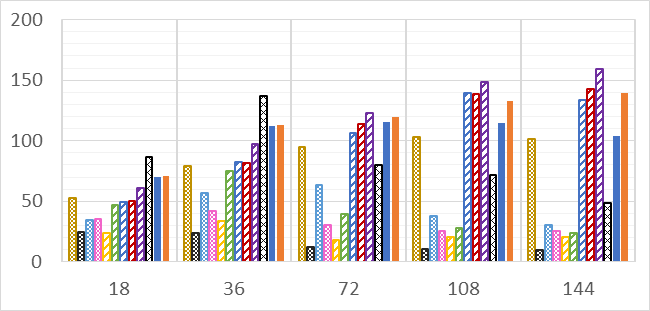}
        \\
        \vspace{-6mm}\rotatebox{90}{5\% updates} &
        \vspace{-6mm}\includegraphics[width=\linewidth]{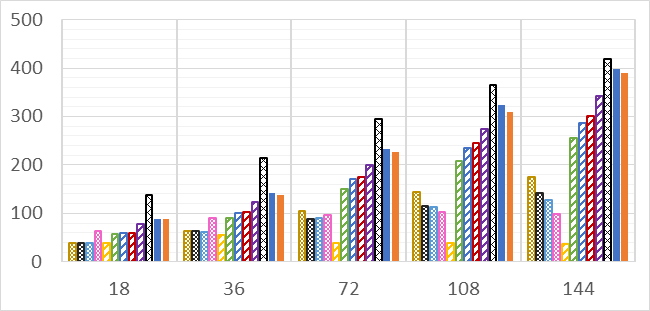} &
        \vspace{-6mm}\includegraphics[width=\linewidth]{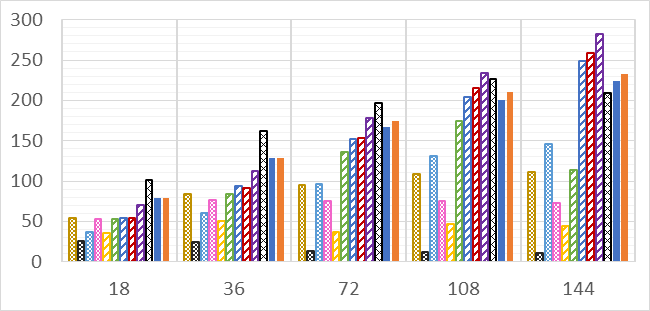}
    \end{tabular}
\end{minipage}
\vspace{-2mm}
\includegraphics[width=0.8\linewidth]{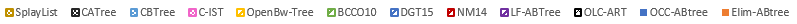}
\vspace{-3mm}
\caption{SetBench microbenchmark with 10K keys. x-axis: number of threads. y-axis: operations per $\mu$s.}
\label{fig-results-20k}
\vspace{-3mm}
\end{figure*}

\begin{figure*}[t]
    \centering
    \setlength\tabcolsep{0pt}
    \vspace{-1mm}
\begin{minipage}{1\linewidth}
    \centering
    \begin{tabular}{m{0.04\linewidth}m{0.48\linewidth}m{0.48\linewidth}}
        &
        \fcolorbox{black!50}{black!20}{\parbox{\dimexpr \linewidth-2\fboxsep-2\fboxrule}{\centering {Zipf parameter = 0 (Uniform)}}} &
        \fcolorbox{black!50}{black!20}{\parbox{\dimexpr \linewidth-2\fboxsep-2\fboxrule}{\centering {Zipf parameter = 1 (Skewed)}}}
        \\
        \rotatebox{90}{100\% updates} &
        \includegraphics[width=\linewidth]{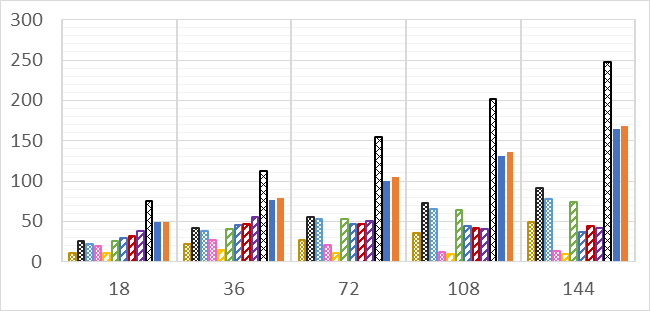} &
        \includegraphics[width=\linewidth]{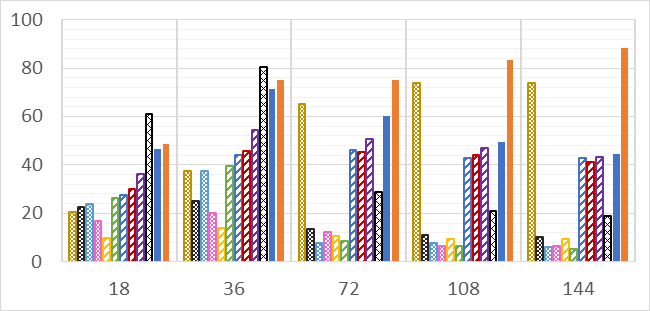}
        \\
        \vspace{-5mm}\rotatebox{90}{50\% updates} &
        \vspace{-6mm}\includegraphics[width=\linewidth]{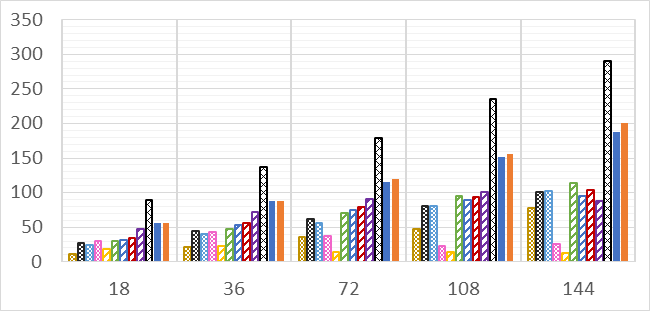} &
        \vspace{-6mm}\includegraphics[width=\linewidth]{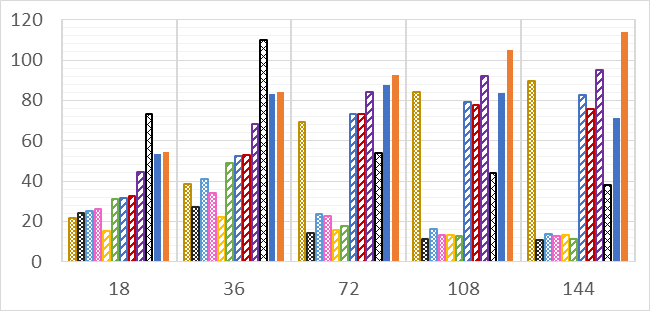}
        \\
        \vspace{-6mm}\rotatebox{90}{20\% updates} &
        \vspace{-6mm}\includegraphics[width=\linewidth]{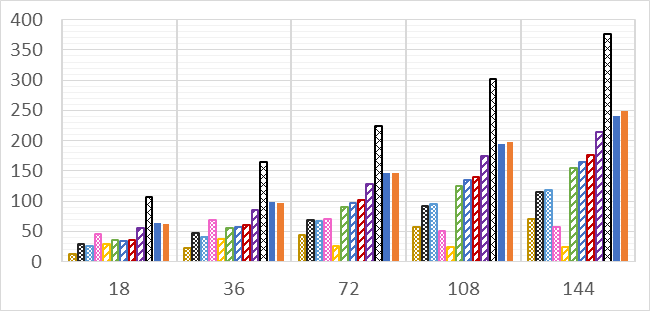} &
        \vspace{-6mm}\includegraphics[width=\linewidth]{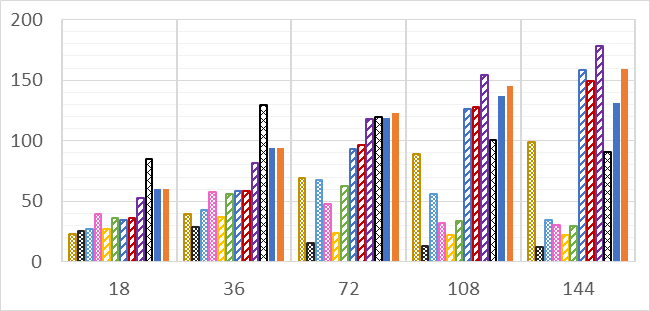}
        \\
        \vspace{-6mm}\rotatebox{90}{5\% updates} &
        \vspace{-6mm}\includegraphics[width=\linewidth]{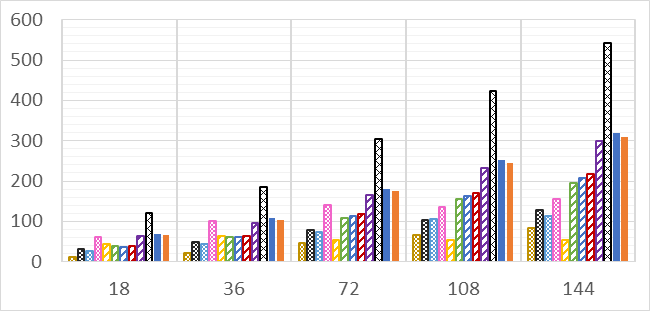} &
        \vspace{-6mm}\includegraphics[width=\linewidth]{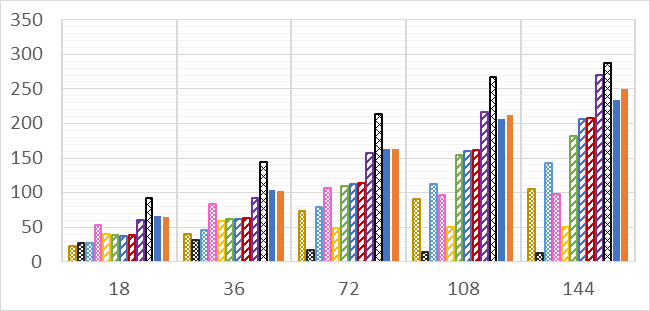}
    \end{tabular}
\end{minipage}
\vspace{-2mm}
\includegraphics[width=0.8\linewidth]{img/legend_full.png}
\vspace{-3mm}
\caption{SetBench microbenchmark with 100K keys. x-axis: number of threads. y-axis: operations per $\mu$s.}
\label{fig-results-200k}
\vspace{-3mm}
\end{figure*}

\begin{figure*}[t]
    \centering
    \setlength\tabcolsep{0pt}
    \vspace{-1mm}
\begin{minipage}{1\linewidth}
    \centering
    \begin{tabular}{m{0.04\linewidth}m{0.48\linewidth}m{0.48\linewidth}}
        &
        \fcolorbox{black!50}{black!20}{\parbox{\dimexpr \linewidth-2\fboxsep-2\fboxrule}{\centering {Zipf parameter = 0 (Uniform)}}} &
        \fcolorbox{black!50}{black!20}{\parbox{\dimexpr \linewidth-2\fboxsep-2\fboxrule}{\centering {Zipf parameter = 1 (Skewed)}}}
        \\
        \rotatebox{90}{100\% updates} &
        \includegraphics[width=\linewidth]{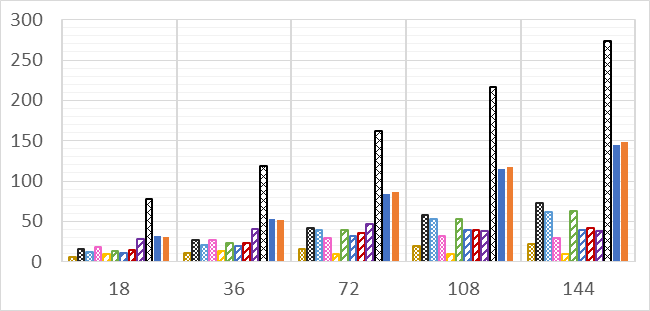} &
        \includegraphics[width=\linewidth]{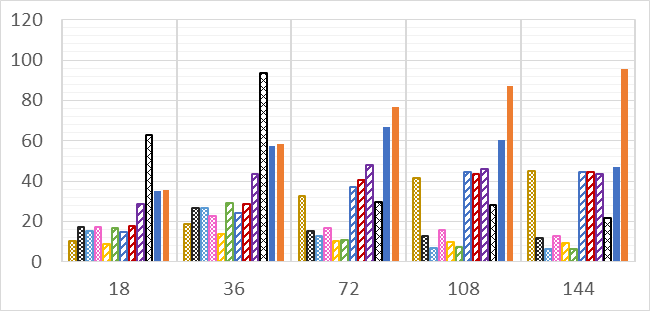}
        \\
        \vspace{-5mm}\rotatebox{90}{50\% updates} &
        \vspace{-6mm}\includegraphics[width=\linewidth]{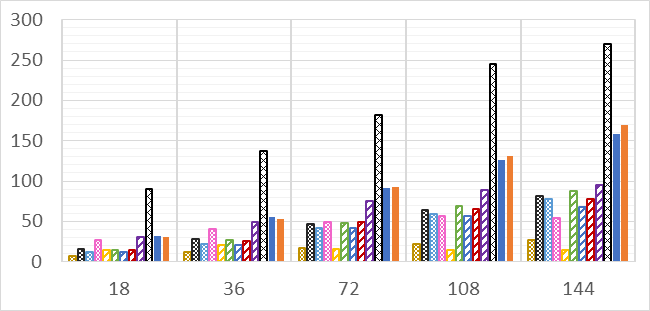} &
        \vspace{-6mm}\includegraphics[width=\linewidth]{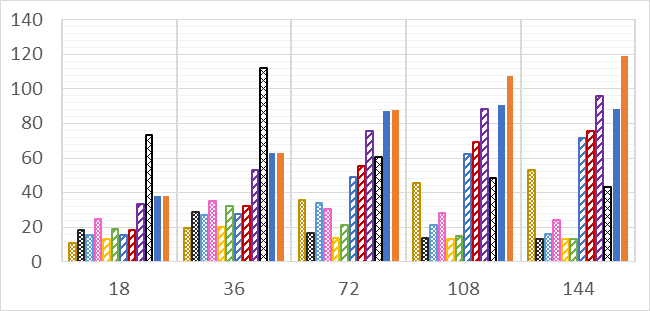}
        \\
        \vspace{-6mm}\rotatebox{90}{20\% updates} &
        \vspace{-6mm}\includegraphics[width=\linewidth]{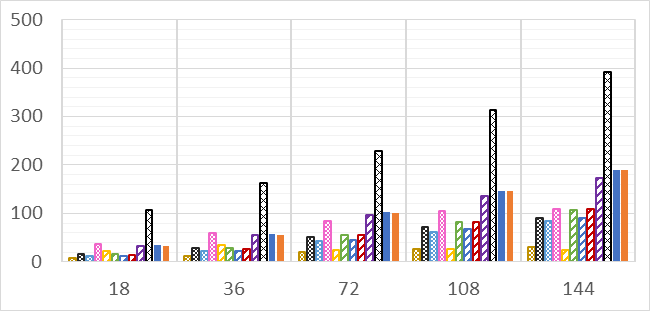} &
        \vspace{-6mm}\includegraphics[width=\linewidth]{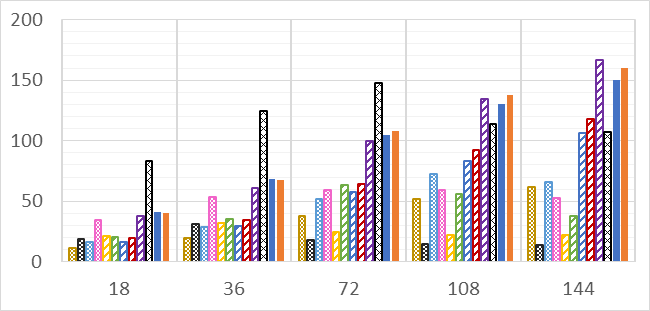}
        \\
        \vspace{-6mm}\rotatebox{90}{5\% updates} &
        \vspace{-6mm}\includegraphics[width=\linewidth]{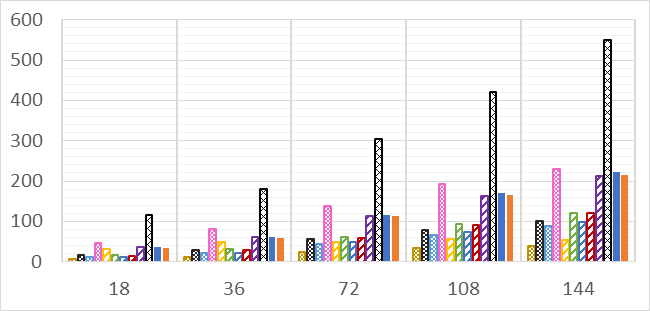} &
        \vspace{-6mm}\includegraphics[width=\linewidth]{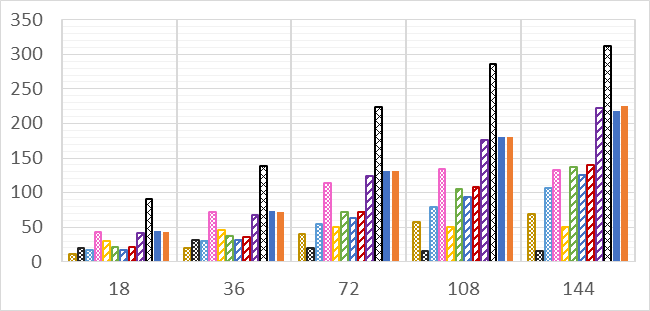}
    \end{tabular}
\end{minipage}
\vspace{-2mm}
\includegraphics[width=0.8\linewidth]{img/legend_full.png}
\vspace{-3mm}
\caption{SetBench microbenchmark with 1M keys. x-axis: number of threads. y-axis: operations per $\mu$s.}
\label{fig-results-2m}
\vspace{-3mm}
\end{figure*}

\begin{figure*}[t]
    \centering
    \setlength\tabcolsep{0pt}
    \vspace{-1mm}
\begin{minipage}{1\linewidth}
    \centering
    \begin{tabular}{m{0.04\linewidth}m{0.48\linewidth}m{0.48\linewidth}}
        &
        \fcolorbox{black!50}{black!20}{\parbox{\dimexpr \linewidth-2\fboxsep-2\fboxrule}{\centering {Zipf parameter = 0 (Uniform)}}} &
        \fcolorbox{black!50}{black!20}{\parbox{\dimexpr \linewidth-2\fboxsep-2\fboxrule}{\centering {Zipf parameter = 1 (Skewed)}}}
        \\
        \rotatebox{90}{100\% updates} &
        \includegraphics[width=\linewidth]{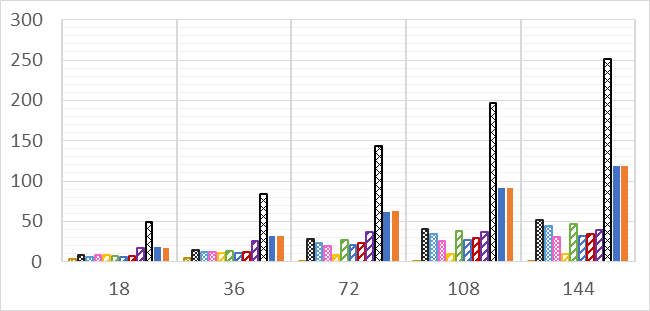} &
        \includegraphics[width=\linewidth]{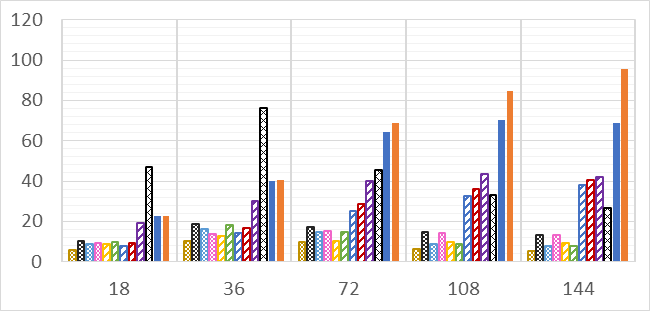}
        \\
        \vspace{-5mm}\rotatebox{90}{50\% updates} &
        \vspace{-6mm}\includegraphics[width=\linewidth]{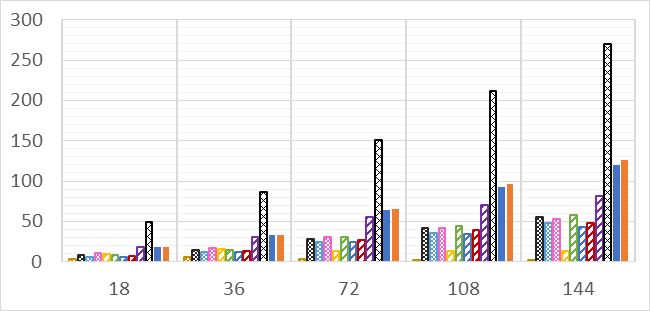} &
        \vspace{-6mm}\includegraphics[width=\linewidth]{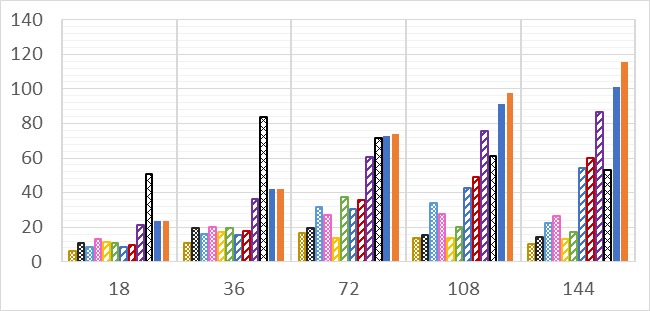}
        \\
        \vspace{-6mm}\rotatebox{90}{20\% updates} &
        \vspace{-6mm}\includegraphics[width=\linewidth]{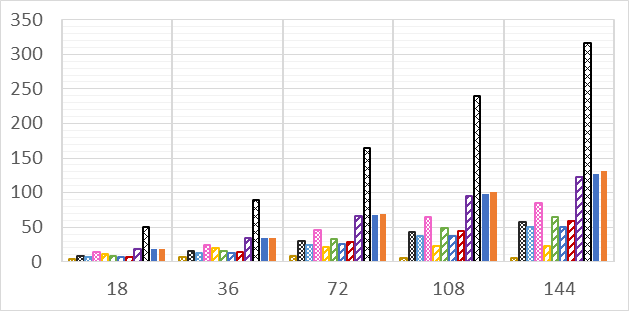} &
        \vspace{-6mm}\includegraphics[width=\linewidth]{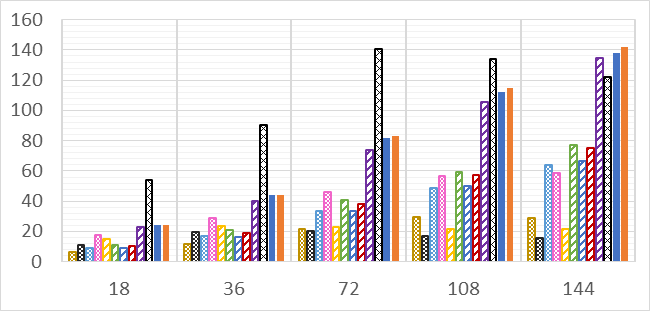}
        \\
        \vspace{-6mm}\rotatebox{90}{5\% updates} &
        \vspace{-6mm}\includegraphics[width=\linewidth]{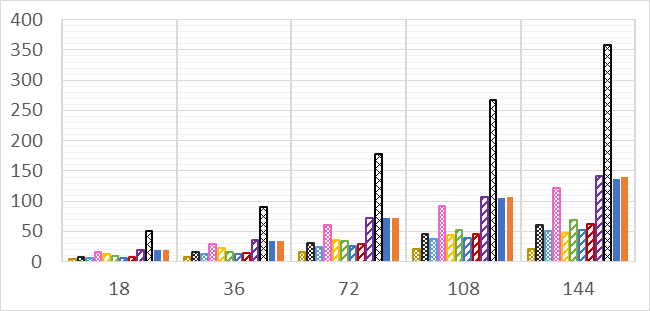} &
        \vspace{-6mm}\includegraphics[width=\linewidth]{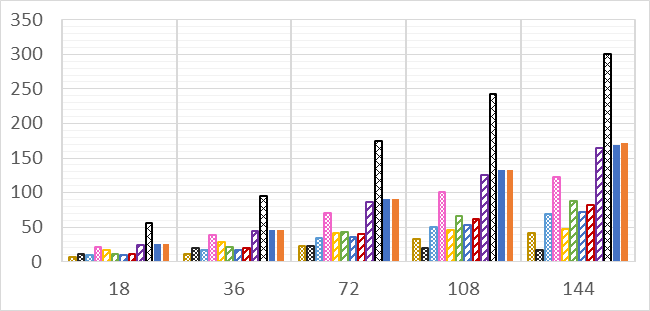}
    \end{tabular}
\end{minipage}
\vspace{-2mm}
\includegraphics[width=0.8\linewidth]{img/legend_full.png}
\vspace{-3mm}
\caption{SetBench microbenchmark with 10M keys. x-axis: number of threads. y-axis: operations per $\mu$s.}
\label{fig-results-20m}
\vspace{-3mm}
\end{figure*}

In this section, we compare our trees with other leading dictionary implementations using both a synthetic microbenchmark and the Yahoo! Cloud Serving Benchmark~\cite{YCSB}, as implemented in SetBench (a framework for benchmarking concurrent dictionaries)~\cite{C-IST}.
% , using the same experimental methodology as~.
% Our implementation, along with instructions to replicate these experiments, is available at the following anonymized repository: \url{https://gitlab.com/anonuser078923/ppopp2022}. %\textbf{Additional experiments appear in the supplementary material}.
%and find that our trees outperform most of them on both low- and high-contention workloads.

\textbf{See Section~\ref{sec:related} for descriptions of the data structures included in our graphs.}
In the following figures, solid bars represent our trees, striped bars represent data structures that are distribution-naïve (\lfab, BCCO10, NM14, DGT15, OLC-ART, OpenBw-Tree), and checkered bars represent data structures that adapt their structure to the access distribution (CATree, CBTree, SplayList), or try to exploit it to obtain faster searches (C-IST).

\vspace{1mm}\mypara{System}
Our volatile memory experiments (Figure~\ref{fig-results-2m}, Figure~\ref{fig:YCSB}) run on a 4-socket Intel Xeon Gold 5220 with 18 cores per socket and 2 hyperthreads (HTs) per core (for a total of 144 hardware threads), and 192GiB of RAM.
Our persistent memory experiments (Section~\ref{sec:persistenceExp}) run on a 2-socket Intel Xeon Gold 5220R CLX with 24 cores per socket and 2 HTs per core (for a total of 96 hardware threads), 192GiB of RAM, and 1536GiB of Intel 3DXPoint NVRAM.
In all of our experiments, we pin threads such that the first socket is saturated before the second socket is used, and so on.
Additionally, the pinning ensures that all cores on a socket are used before hyperthreading was engaged.
The machine runs Ubuntu 20.04.2 LTS.
All code is written in C++ and compiled with G++ 7.5.0-3 with compilation options \var{-std=c++17 -O3}.
We use the scalable allocator jemalloc 5.0.1-25.
We use \var{numactl -i all} to interleave pages evenly across all NUMA nodes.

\vspace{1mm}\mypara{Memory reclamation}
All data structures use DEBRA, a variant of epoch-based memory reclamation~\cite{DEBRA}, except the SplayList and FPTree (which do not reclaim memory) and the OpenBw-Tree (which uses a different epoch-based reclamation scheme which we were unable to change due to its complexity).

\vspace{1mm}\mypara{Methodology}
Each experiment \textit{run} starts with a prefilling phase, in which a random subset of 8-byte keys and values are inserted into the data structure until the data structure size reaches its expected steady-state size (half of the key range, since the proportions of inserts and deletes are equal in our experiments).
After the prefilling phase, $n$ threads are created and started together, and the \textit{measured} phase of the experiment begins.
In this phase, each thread repeatedly selects an operation (\ins, \del, \find) based on the desired update frequency, and selects a key according to a uniform or Zipfian distribution. %, and threads doing an insert select a uniformly random key to inser.
This continues for 10 seconds, and the total \textit{throughput} (operations completed per second) is recorded.
Each experiment is run three times, and our graphs plot the averages of these runs.

\vspace{1mm}\mypara{Validation}
To sanity-check the correctness of the evaluated data structures, each thread keeps track of the sum of keys that it successfully inserts and deletes.
At the end of each run, all threads' sums are added to a grand total, and the grand total must match the sum of keys in the data structure.

\vspace{-2mm}\subsection{SetBench microbenchmark}
\vspace{-1mm}\mypara{Read-mostly (5\% updates)}
Performance on read-mostly workloads has been shown to be correlated with short paths to keys, since shorter paths resulted in fewer cache misses (which dominate runtime in read-mostly workloads)~\cite{C-IST}.
% The C-IST has nodes that are much larger than those of the (a,b)-trees and the OpenBw-Tree, and quickly traverses them using interpolation search.
% For example, the root in the C-IST contains $O(\sqrt{n})$ keys, where $n$ is the size of the entire tree.
% As a result, the C-IST is very shallow and has extremely fast searches.
Thus, we expected the (a,b)-trees, OpenBw-Tree, CBTree, and C-IST (all of which use fat nodes containing many pointers) to be the fastest. However, this is only true for the (a,b)-trees.
The C-IST, which is heavily optimized for search-only workloads, performs well in the uniform case, but performs much worse in the Zipfian case.
The OpenBw-Tree performs poorly in both workloads. However, a short experiment suggests that both the C-IST and OpenBw-Tree perform comparably to the (a,b)-trees with \textit{no} updates. The extent to which just 5\% updates affects their read performance is surprising.
The BSTs (BCCO10, NM14) have similar performance relative to one another (roughly half that of the (a,b)-trees).

The CBTree and SplayList fell short of our expectations on the Zipfian workload.
We expected that splaying would greatly accelerate searches (especially since the splayed key is never removed in a read-mostly workload), but they barely exceed their performance on the uniform workload.
The CATree's performance is reasonable on the uniform workload, but is much worse than the other data structures on the Zipfian workload.
%This observation is true in the workloads with updates, as well, %This trend holds for the two updating workloads as well,
%since 
All of the CATree's operations (even searches) require \textit{locking} a leaf. % that is being searched or modified. 
% The CATree's authors describe a lock-free CATree but we could not find a high performance C/C++ implementation, and according to their own results, their optimized implementation does not scale beyond 16 threads.
% Even assuming the same (approx. 2x) performance benefit they observe from the optimization, the lock-free optimized CATree would still be slower than our trees.

% The poor performance of the CBTree, SplayList, and CATree illustrates a common flaw of distribution-adaptive data structures: they rely on carefully tuned parameters that determine how often adjustments should be performed. %, and these parameters may not be optimal for all systems (or thread counts).
% We used the same parameter settings used by the authors in their own experiments, and could not see how to easily improve them.

% Our implementation of the CATree is a port of the authors' Java code. The authors also describe a lock-free CATree and an optimization to the CATree that makes it more performant in high-contention settings by making some leaf nodes lock-free, but we were unable to port these versions of their Java code to C++ because of their complexity. However, the performance of our implementation of the locking CATree seems to match the performance in their paper. Per their own results, the optimized version does not scale past 16 threads, and even assuming the same (~2x) performance benefit they observe from the optimization, the lock-free optimized CATree would still likely be slower than our trees.

\vspace{3mm}\mypara{Update-heavy (50\%, 100\% updates)}
Overall, throughput decreases as the proportion of updates increases (as expected).

On \textit{uniform} update-heavy workloads, the \lfab\ and the C-IST scale much worse than our trees.
The \lfab\ creates a new copy of a (fat) node every time a key is inserted.
The C-IST must completely rebuild the tree after $n/4$ updates, where $n$ is the size of the tree.
As a result, both incur high overhead for updates. %This is because both the \lfab\ and the C-IST trees require changes to the tree \textit{structure} whenever a key is inserted or removed. %, causing tremendous overhead.
%This paper has attempted to address this issue with the use of fine-grained locking and the results seem to indicate that we have succeeded. % remove maybe?
%The other data structures scale about as well as our trees, except NM14 which scales slightly worse.
The other competing trees have better scaling but relatively poor absolute throughput.
Our trees are roughly 2x faster than the leading competitor (the CATree) in the uniform 100\% workload.

On \textit{skewed} update-heavy workloads, the benefit of publishing elimination becomes clear.
The \elim\ is significantly faster than the \occ\ on these workloads, with the gap increasing as the proportion of updates does.
At 100\% updates, the \elim\ is up to 2.5x as fast as its fastest competitor.
The C-IST still scales poorly on these workloads, but the \lfab\ performs extremely well, outperforming even the \occ\ at 50\% updates.
At relatively low update rates, the benefit of lock-freedom (i.e., faster threads \textit{helping} slower threads) exceeds the overhead of allocating new nodes for each key inserted.
At the highest update rates, the overhead of managing memory dominates the performance of the \lfab.

NM14 scales much better than BCCO10 in these workloads, slightly exceeding the performance of the \occ.
This is because searches in BCCO10 have to restart many times because of frequent updates along the path to the frequently-accessed keys.
% And because NM14 just marks/unmarks to delete/insert
% The \elim\ significantly outperforms the other data structures in the 100\% Zipfian update workload.
% Even with only 20\% updates, the Zipfian distribution induces enough contention for the \elim\ to match the \lfab's performance.
% Data structures that perform much worse in our update workloads can be divided into two groups.
% The first group is comprised of data structures that have slow updates even when keys are uniformly distributed.
% This group includes the \lfab, OpenBw-Tree, and the C-IST, all of which perform much worse on 100\% uniform updates than in the read-mostly workload, especially when running on multiple processor sockets (72 threads and above).
% The \lfab\ handles a moderate amount of updates well (outperforming all non-eliminating data structures at 144 threads), but the OpenBw-Tree and C-IST do not.
% In the case of the C-IST, this is because rebuilding the tree after updates is extremely costly.
% The second group is comprised of data structures whose updates do not scale well with contention (i.e., in the Zipfian workload).
% This group includes the \occ\ (for the reasons mentioned in Section~\ref{sec:occ_ABtree}), the CATree (for the reasons mentioned above), and BCCO10, which uses OCC and might suffer from the same problems that affect the \occ.
A notable outlier in the skewed update-heavy workloads is the SplayList, which had relatively poor read-mostly performance but matches the performance of NM14 and the \lfab\ on the skewed update-only workload.
This may be partially because the SplayList never frees memory (simply marking keys as deleted instead), so reinserting a key that was once in the SplayList requires no memory allocation (which normally adds considerable overhead to the other data structures).
This approach is quite efficient in our microbenchmark, but might be less so if the set of keys that are \textit{ever} inserted is much larger than the set of keys that are \textit{typically} in the dictionary.

\begin{figure}
\centering
\includegraphics[width=0.46\textwidth]{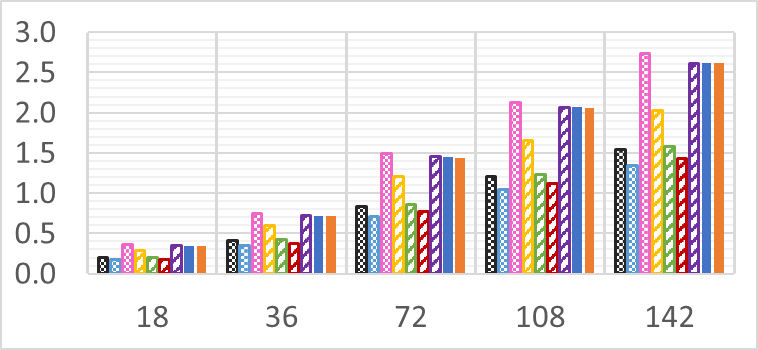}
\vspace{-2mm}
\includegraphics[width=0.9\linewidth]{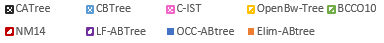}
\vspace{-2mm}
\caption{YCSB throughput on Workload A. x-axis: number of threads. y-axis: transactions per $\mu$s.}
\label{fig:YCSB}
\vspace{-2mm}
\end{figure}

\vspace{-2mm}\subsection{YCSB}
The Yahoo! Cloud Serving Benchmark (YCSB) is a standard tool for benchmarking concurrent database indices~\cite{YCSB}. We run the benchmark using the above data structures as the database index.
We run Workload A (50\% reads, 50\% writes, Zipf factor 0.5) from the YCSB standard workloads, with a uniform access distribution and an initial data structure size of 100M (Figure~\ref{fig:YCSB}).
Figure~\ref{fig:YCSB} does not contain the SplayList since it does not reclaim memory and consequently caused the system to run out of memory.
Note that the writes in the YCSB workload are to the database itself, not the index. 
That is, a YCSB write simply \textit{reads} the row pointer from the index, then locks the row, updates it, and unlocks it (without modifying the index).
As a result, the results are closest to our microbenchmark uniform read-mostly workload.

\subsection{Persistence experiments}
\label{sec:persistenceExp}
Of the concurrent persistent trees in Section~\ref{sec:related}, only the FPTree and RNTree have publicly available implementations that passed our validation scheme (both implementions were from~\cite{RNTree}).
However, these implementations do not reclaim memory.

Figure~\ref{fig:FPTree-comparisons} shows the results of our microbenchmark on our persistent memory machine.
Even with the overhead of reclaiming memory, the \pmocc\ and \pmelim\ outperform both the FPTree and the RNTree on all thread counts. (Results on smaller/larger key ranges and different update percentages were similar).
In the uniform case, the FPTree performs similarly to our trees at low thread counts but exhibits extreme negative scaling when running on 2 sockets (96 threads).
However, this might be an artifact of this particular implementation, since the original paper shows better scaling on 2 sockets.
The RNTree performs worse than the FPTree on uniform workloads, but slightly better on the Zipfian workload.
Both the FPTree and RNTree also exhibit negative scaling in the Zipfian case, even when running on only one socket.

\begin{figure}[t]
\centering
\includegraphics[width=\linewidth]{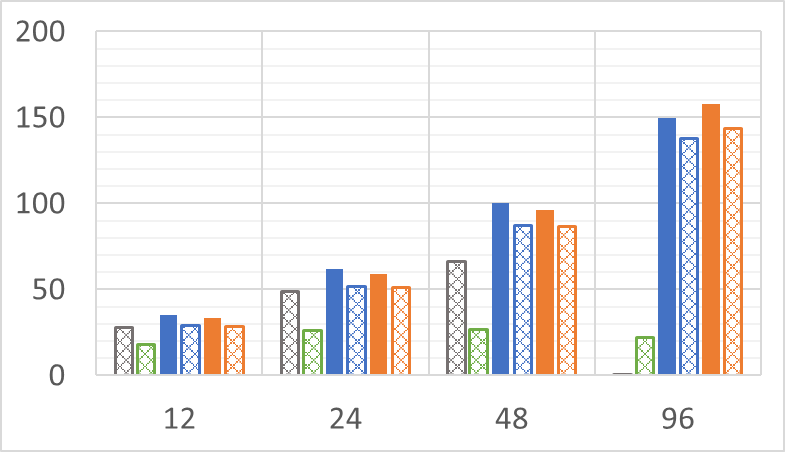}

%\vspace{-6mm}
\includegraphics[width=\linewidth]{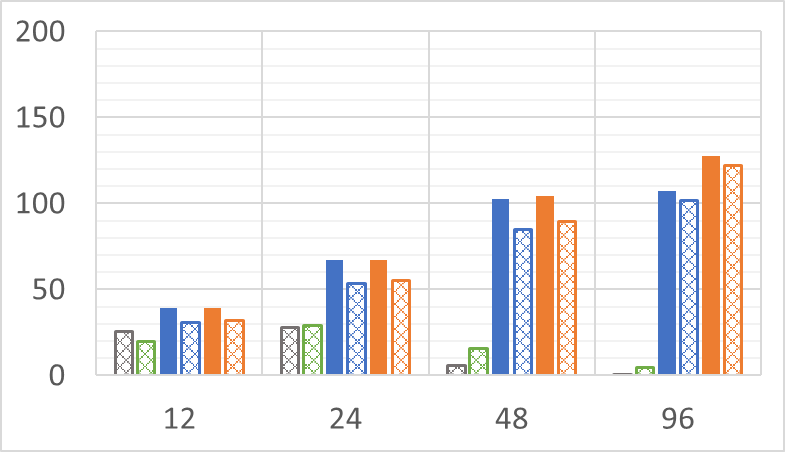}

\includegraphics[width=\linewidth]{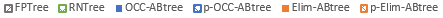}
\vspace{-5mm}
\caption{Comparing with other persistent trees: SetBench microbenchmark with 1M keys, 50\% updates (25\% insert and 25\% delete). Left: Uniform access distribution. Right: Zipfian access distribution (with Zipf factor 1). x-axis: number of threads. y-axis: operations per $\mu$s.}
\label{fig:FPTree-comparisons}
\vspace{-3mm}
\end{figure}

We attempted to compare with an unofficial implementation of the BzTree~\cite{PiBench}, but encountered failures during validation. The implementors mentioned that the errors might be fixable, but were unable to produce a fix in time for this publication.
Table~\ref{tab:persistence} shows the persistence overhead of our trees.
Comparing with the overheads listed in the BzTree paper, the overhead of our trees is slightly less than the BzTree's average persistence cost of 5\% on a uniform 10\%-update workload and 12\% on a uniform 50\%-update workload.
% The update-only uniform workload is the only one in which our trees experience a large drop in performance (because of the large number of updates, which cannot be eliminated).
% As expected, elimination reduces persistence overhead in the Zipfian update workloads. %the \elim\ has less persistence overhead than the \occ\ in Zipfian update workloads. %because elimination causes it to perform fewer updates (and hence fewer flushes).
\begin{table}[h]
    \centering
    \begin{tabular}{|c|c|c|c|c|c|c|}
    \hline
    & \multicolumn{3}{c|}{Uniform} & \multicolumn{3}{c|}{Zipfian} \\
    \hline
    \textbf{Update rate:} & \textbf{100\%} & \textbf{50\%} & \textbf{10\%} & \textbf{100\%} & \textbf{50\%} & \textbf{10\%} \\
    \hline
    \pmocc\  & -16\% & -8\% & -6\% & -6\% & -9\% & -7\% \\
    \hline
    \pmelim\ & -14\% & -9\% & -1\% & -5\% & -5\% & -5\%\\
    \hline
    \end{tabular}
    \caption{Change in throughput upon enabling persistence. 96 threads, 1 million keys.} %non-volatile RAM as the main memory.}
    \label{tab:persistence}
\end{table}

%\vspace{-7mm}
\section{Future work and conclusion}
It would be interesting to explore the interaction between publishing elimination and different data structure semantics.
Publishing elimination remains correct for some alternative definitions of insert.
If insert replaces existing keys but returns no value (instead of simply returning the existing key), publishing elimination does not require any modifications: the thread that successfully modifies the data structure is linearized last.

On the other hand, if insert returns the value it replaces, then publishing elimination would require changes to allow each insert in a sequence of linearized inserts to communicate its value to the next insert.
%changes to work, as this requires each thread to communicate the value it inserted to the thread that is linearized after it.
%if inserts needthe semantics of \insargs\ are slightly different, in particular, if inserts should \textit{replace} but not return any existing value,
% then publishing elimination is still a also works if the desired semantics of  (i.e., the thread that successfully inserts is linearized last). However, publishing elimination \textit{would require changes to work :)} does not work if the replaced value must be returned, as this requires each thread to communicate the value it inserted to the thread that is linearized after it.
% It would interesting to explore methods to accommodate this.

% There are also ways to eliminate operations without using a queue. One might use a queue of short arrays instead, and the thread eliminating operations would eliminate only those in its array. This would give some cache locality and thus reduce the amount of cache misses during elimination.

Using MCS locks (instead of test-and-test-and-set spinlocks) significantly increased the scalability of the \occ. Using NUMA-aware locks like HCLH~\cite{HCLH}, lock cohorting~\cite{LockCohorting}, or NUMA-aware reader-writer locks~\cite{NumaReaderWriter} might also be a simple way of improving performance further.
% The \elim\ might also benefit from more scalable locks, though using queue locks using \var{tryLock} with them defeats the point of the queue.

We have introduced the \occ, which provides good performance in both read-mostly and update-heavy workloads, and the \elim\ which uses publishing elimination to further improve performance in high-contention workloads.
Finally, we have presented persistent versions of our trees that require only minor modifications and are still highly performant.

% \section{Acknowledgements}

\begin{acks}
This work was supported by: the Natural Sciences and Engineering Research Council of Canada (NSERC) Collaborative Research and Development grant: CRDPJ 539431-19, the Canada Foundation for Innovation John R. Evans Leaders Fund with equal support from the Ontario Research Fund CFI Leaders Opportunity Fund: 38512, Waterloo Huawei Joint Innovation Lab project ``Scalable Infrastructure for Next Generation Data Management Systems'', NSERC Discovery Launch Supplement: DGECR-2019-00048, NSERC Discovery Program under the grants: RGPIN-2019-04227 and RGPIN-04512-2018, and the University of Waterloo. We would also like to thank the reviewers for their insightful comments. % that helped us to improve the manuscript.
\end{acks}
% None.

\bibliographystyle{ACM-Reference-Format}
\bibliography{references}

\section{Artifact Description}
The artifact containing the source code for all algorithms and experiments run in this paper is available at \url{https://doi.org/10.5281/zenodo.5733351}.

Note: Sudo permission may be required to execute the following instructions.
\begin{enumerate}
\item Install the latest version of Docker on your system.
The artifact was tested with the Docker version 20.10.2. (Instructions to install Docker can
be found at
\url{https://docs.docker.com/get-docker/}.)

\item Download the artifact from Zenodo at URL:
\url{https://doi.org/10.5281/zenodo.5733351}.

\item Load the downloaded
docker image:\\
\texttt{\$ sudo docker load -i setbench.tar.gz}

\item Verify that image was loaded:\\
\texttt{\$ sudo docker images}

\item Start a docker container from the loaded image:\\
\texttt{\$ sudo docker run -p 2222:22 -d --privileged --name setbench setbench}

\item Verify that the container is running (you should see a setbench container):\\
\texttt{\$ sudo docker container ls}

\item SSH into the running container with password \texttt{root}:\\
\texttt{\$ ssh root@localhost -p 2222}

\item Follow the instructions in \texttt{setbench/README.md} to replicate results. Note that you might have to change thread counts in the \texttt{run.sh} and \texttt{run\_persistence\_cost.sh} scripts to match the constraints of your system.
\end{enumerate}

\end{document}